\documentclass[twocolumn,pre,superscriptaddress]{revtex4-2}
\usepackage[utf8]{inputenc}
\usepackage{xcolor}
\usepackage{babel,comment}
\usepackage{dsfont}
\usepackage{amscd}
\usepackage{amsmath}
\usepackage{amsthm}
\usepackage{amssymb}
\usepackage{bm}
\usepackage{units}
\usepackage{cancel}
\usepackage{graphicx}
\usepackage{upgreek}

\usepackage[colorlinks=true]{hyperref}

\newtheorem{theorem}{Theorem}

\newtheorem{lemma}[theorem]{Lemma}
\newtheorem{proposition}[theorem]{Proposition}
\theoremstyle{remark}
\newtheorem{remark}{Remark}
\theoremstyle{definition}
\newtheorem{definition}{Definition}
\newtheorem{example}{Example}

\begin{document}
\title{%
Covariant influences for finite discrete dynamical systems}
\author{Carlo Maria Scandolo}
\email{carlomaria.scandolo@ucalgary.ca}
\affiliation{Department of Mathematics \& Statistics, University of Calgary, Calgary,
AB,  T2N 1N4, Canada}
\affiliation{%
Institute for Quantum Science and Technology, University of Calgary,
Calgary, AB, T2N 1N4, Canada}
\author{Gilad Gour}
\email{gour@ucalgary.ca}
\affiliation{Department of Mathematics \& Statistics, University of Calgary, Calgary,
AB,  T2N 1N4, Canada}
\affiliation{%
Institute for Quantum Science and Technology, University of Calgary,
Calgary, AB,  T2N 1N4, Canada}
\author{Barry C.\ Sanders}
\email{sandersb@ucalgary.ca}
\affiliation{%
Institute for Quantum Science and Technology, University of Calgary,
Calgary, AB,  T2N 1N4, Canada}
\begin{abstract}
We develop a rigorous theory of external influences on finite discrete dynamical systems, going beyond the perturbation paradigm, in that the external influence need not be a small contribution.
Indeed, the covariance condition can be stated as follows: if we evolve the dynamical system for $n$ time steps and then we disturb it, it is the same as first disturbing the system with the same influence and then letting the system evolve for $n$ time steps.
Applying the powerful machinery of resource theories, we develop a theory of covariant influences both when there is a purely deterministic evolution and when randomness is involved.
Subsequently, we provide necessary and sufficient conditions for the transition between states under deterministic covariant influences and necessary conditions in the presence of stochastic covariant influences, predicting which transitions between states are forbidden.
Our approach, for the first time, employs the framework of resource theories, borrowed from quantum information theory, to the study of  finite discrete dynamical systems.
The laws we articulate unify the behavior of different types of  finite discrete dynamical systems, and their mathematical flavor makes them rigorous and checkable.
\end{abstract}
\maketitle
\tableofcontents
\section{Introduction}
\label{sec:intro}
Dynamical systems describe the evolution of several interesting situations.  There are essentially two flavors: continuous dynamical systems, which are more often studied, and discrete dynamical systems, which are the subject here.
In many cases, especially for continuous dynamical systems, when the evolution is particularly complex to deal with, one splits the evolution into two parts: the uninfluenced part and the perturbation.
The latter is interpreted as a small correction to the uninfluenced evolution. In this article, we go beyond the perturbation paradigm by introducing the notion of \emph{covariant influence}, which need not be a small contribution. Such a splitting can be done on the basis of time scale: we can think of the covariant influence as acting on a significantly longer time scale than the uninfluenced evolution.
A small influence is a perturbation, so to get a more interesting behavior we must also go beyond small  influences. The covariance condition guarantees that, despite not being small, it preserves the underlying structure representing the evolution of states. In this way, the evolution of a discrete dynamical system can  be written as the composition of two evolutions: one is understood as the basic evolution of the system, and the other is the covariant influence.
The covariance requirement ensures that the order in which these two parts are applied does not matter.

In this article, for the first time, we develop a general theory of covariant influences in discrete dynamical systems with a finite number of states. Such a theory comes in two flavors. The first is a deterministic one, where randomness is completely forbidden both in the initial state and in the action of the covariant influence. In the second, instead, we allow the presence of randomness both in the initial state and in the covariant influence. We show that  deterministic influences allow hopping between attractors whose length becomes smaller and smaller. Instead, in the presence of randomness all jumps between attractors become possible. In particular, we achieve a full characterisation of transitions between states in the deterministic setting, and in the random case, we predict which transitions between states are forbidden.

Our approach, for the first time, employs the framework of resource theories, borrowed from quantum information theory, to the study of discrete dynamical systems. This also constitutes the first application of resource theories outside the physics domain, to a field with countless applications to diverse areas of science, including genetic regulatory networks. The advantage of the resource-theoretic approach is the separation of the evolution of a discrete dynamical system into an ``uninfluenced'' contribution and an influence, as is standard practice when studying perturbations. However, the influence need not be a perturbative (i.e.\ small) contribution~\cite{Arnold,Fasano-Marmi}.

Our results for discrete dynamical systems are obtained from first principles that are of a mathematical flavor, as they are grounded in category theory. This ensures that our analysis is rigorous, logical, and checkable. In particular, the use of resource theories allows us to get a unified picture of discrete dynamical systems under influences, regardless of the specifics of their evolutions, unlike most standard approaches to discrete dynamical systems. In this way, our results, in the form of simple mathematical laws, can be phrased in general terms, so they are applicable to a broad class of discrete dynamical systems. The key concept in our analysis is covariance, which can be thought of as a symmetry in time evolution, and can be expressed by a simple commutativity condition.
The general laws we find are predictive for discrete dynamical systems where covariance is in force, and motivate novel experimental work.

The article is organized as follows: in~\S\ref{sec:background}
we present the background of this work, starting with a presentation of discrete dynamical systems, and ending with a presentation of resource theories.
In~\S\ref{sec:approach}, we introduce the covariance condition, showing that it gives rise to a well-behaved resource theory. In~\S\ref{sec:results}, we present our results both for deterministic dynamical systems and dynamical systems with randomness. A summary of the results is presented in~\S\ref{sec:discussion}, and conclusions are drawn in~\S\ref{sec:conclusions}.
\section{Background}
\label{sec:background}
In this section we provide key background elements that underpin our contributions and advances in subsequent sections.
We begin in~\S\ref{subsec:DDS} with an overview of a discrete dynamical system,
which is a set of elements that evolve discretely in time.
Then we present random Boolean networks as a special case of discrete dynamical systems in~\S\ref{subsec:Random-Boolean-networks}.
The subsequent subsection~\ref{subsec:resource} provides essential elements of resource theories.
\subsection{Dynamical systems}
\label{subsec:DDS}

In this subsection we elaborate on essential background concerning discrete dynamical systems.
We begin with explaining general notions of discrete dynamical systems including the set of states and the mapping from states to states corresponding to discrete evolution.
Then we discuss features of discrete dynamical systems such as attractors and basins of attraction,
which arise in studies of dynamics.
A discrete dynamical system can be deterministic or stochastic,
with stochasticity in this case pertaining to states as the dynamics is always deterministic.
\subsubsection{Introduction}
\label{subsubsec:dynamicalnetwork}

Now we introduce the basics of discrete dynamical systems.
Specifically,
we define dynamical systems formally.
Then we present a graphical representation of the dynamical system in the form of a graph.

A dynamical system is a pair $\left(S,\phi_t\right)$ comprising a set~$S$ of elements, with each element called a state~$s$ of the dynamical system,
and a family of functions~$\left\{\phi_{t}\right\}$,
which maps~$S$ to itself, and describe the evolution at time $t$.
If the label~$t$ is in $\mathbb{N}$ or $\mathbb{Z}$, the dynamical system is called \emph{discrete}; if~$t$ is in $\mathbb{R}$, the dynamical system is called \emph{continuous}. The standard requirements for such a family of functions~$\left\{\phi_{t}\right\}$ are the following:
\begin{enumerate}
    \item $\phi_0=\mathds1$, where $\mathds1$ denotes the identity on $S$;
    \item $\phi_{t+s}=\phi_t \circ \phi_s = \phi_s \circ \phi_t$;
    \item $\phi_{-t}=\phi_t^{-1}$.
\end{enumerate}
In this article, we focus on discrete dynamical systems, where time runs in $\mathbb{N}$, and the number of states is finite. Therefore, in this work, whenever we talk about discrete dynamical systems, we always mean \emph{finite} discrete dynamical systems.
For discrete dynamical systems,
we  have discrete time steps, and we define $\phi$ to be the function on~$s$ describing the evolution of~$s$ to itself in  a single time step.
With this in mind, we have
\begin{equation}
\label{eq:phin}
\phi_n:=\phi^n,
\end{equation}
where
$\phi^n$ is the composition of~$\phi$ with itself $n\in\mathbb{N}$ times and refers to the evolution of the system over $n$ time steps: the evolution of a state $s$
after $n$ time-steps is given by 
$\phi^n\left(s\right)$. For this reason we say that $\phi$ is the \emph{generator} of
the dynamics, and in this case we denote the dynamical system simply by $\left(S,\phi\right)$.
By convention, $\phi^0=\mathds1$. From the second of the conditions above, note that evolving the system for $n$ time steps, and then for other~$m$ time steps,
is the same as evolving it for $n+m$ time steps.
Therefore, the set $\left\{\phi^n\right\}_{n\in\mathbb{N}}$ of the powers of $\phi$~(\ref{eq:phin})
has the algebraic structure of a commutative \emph{monoid}, but, in
general, \emph{not} of a group
because monoids are more general structures that include groups, which allow non-invertible evolution.

In many dynamical systems there are external influences that can change the evolution of the system. Think, e.g., of perturbation theory, where we view the perturbation as a small correction to  the evolution of the system~\cite{Arnold,Thirring,Fasano-Marmi}. If we remove the requirement of the perturbation to be small, we have an influence. In mathematical terms, 
an \emph{influence} is simply a mapping 
\begin{equation}
\label{eq:fStoS}
f:S\to S.
\end{equation}
Influences 
form a monoid under function composition (the identity is trivially an influence),
but it is different from the dynamical monoid because it is not generated by a single influence, and as such it is not commutative.
Note that we take $f$ to be in a monoid and not a group because we want to allow  general behavior.
A dynamical graph is a representation of a dynamical system.
A directed graph~$G$ comprises an ordered set of~$M$ vertices~$S$ (also ``nodes''),
and directed edges~$E$,
which are ordered pairs~$\left(s,s'\right)\in S\times S$. The dynamical graph
represents the dynamics corresponding to the evolution $\phi$: 
there is a directed edge from a vertex
$s$ to a vertex~$s'$ if $s'=\phi\left(s\right)$. 
As the dynamics are deterministic, there is only one outgoing edge
for each vertex of the dynamical graph. An example of a dynamical graph is depicted in Fig.~\ref{fig:dynamicalgraph}.
\begin{figure}
\includegraphics[width=\columnwidth]{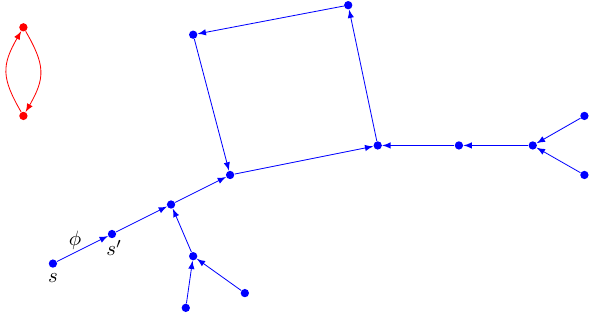}
\caption{%
Dynamical graph:
vertices are states of the dynamical system.
There is an arrow from a state~$s$ to a state $s'$ if $s'=\phi\left(s\right)$. Notice the presence of a cycle of length 2 on the left and a cycle of order 4 on the right. The two connected components of the dynamical graph correspond to the two basins of attraction. The red basin is the set of all states that end up in the cycle of length 2, and the blue basin is the set of all states that end up in the cycle of length 4. Notice that the two basins of attraction are disjoint.}

\label{fig:dynamicalgraph}
\end{figure}

\subsubsection{Features}
\label{subsubsec:features}

We now describe features including cycles, attractors,
transient states and basins of attractions,
which are the fundamental features of a dynamical system.
First we formally define each of these three features.
Then we explain how these features are represented in the dynamical graph,
which represents the dynamical system.

As the dynamics is deterministic on a \emph{finite}
number of states, at some point the dynamics must repeat itself.
Therefore, there is at least one \emph{cycle}.
More formally, a cycle $C$ is a set of states $\{s\}$ such that there exists a positive integer~$\ell$ such that, for each $s\in C$, we have  $\phi^{\ell}\left(s\right)=s$. The minimum integer~$\ell$ is called the \emph{period} or \emph{length} of the cycle. Note that the length of the cycle is also the number of states in the cycle. 
In general, in a dynamical system there are a number of cycles of different periods. Those cycles of period
1 are called \emph{fixed points}. In this case,
\begin{equation}\label{eq:fixed point}
    \phi\left(s\right)=s.
\end{equation}
Cycles are loops in the dynamical graph such that the loop is oriented by directed edges pointing in one way.

An attractor is a set of states that is forward-invariant under the evolution $\phi$. In other words, it is a set of states where evolution ``repeats itself''. As we are working with a finite dynamical system, attractors are actually cycles.  Thus,  all states end up in some cycle.
Some states do not belong to any
cycles; they are called \emph{transient} because it is not possible
to come back to them, not even after enough time steps.

In the dynamical graph, transient states are the vertices of the non-cyclic parts of the graph,
i.e., everything that is not in a loop.
Given an attractor, its \emph{basin of attraction} is the set of all states that end up in that attractor after enough time steps. Notice that, since the dynamics
is deterministic, basins of attractions relative to different attractors are disjoint, i.e., there
can be no state belonging to more than one basin of attraction. In
other words, the set of states~$s$ is partitioned into basins of attraction.
The dynamical graph is made of several connected components: each of them
represents a basin of attraction.
With a little abuse of terminology, we  often say that a state,
even a transient one, has length~$\ell$ if~$\ell$ is the length
of the attractor in its basin of attraction
with the notion of basins made clear in Fig.~\ref{fig:dynamicalgraph}.

\subsubsection{Introducing randomness}
\label{subsubsec:introducingrandomness}

Now we explain how randomness is introduced into the dynamical system.
In some situations, the influence is best described as a stochastic process that activates jumps between the states of a dynamical system according to a certain probability distribution. In such a stochastic setting, we are no longer sure about the state of the dynamical system, which demands the introduction of randomness in the sense of a probability distribution over the states of the dynamical system. Despite randomness, evolution is still represented as a deterministic linear map over the set of state probability distributions.

We now represent randomness on states~$s\in S$
by constructing the $\sigma$-algebra of the power set~$\wp(S)$.
In this case, 
 we first fix an ordering on the elements of~$s$.
 To represent probability measures, it is enough to consider probability
vectors $\bm{p}\in\mathbf{P}^M(\mathbb{R})$, whose entries represent the probability assigned to each of the~$M$ singleton subsets of $S$, each corresponding to one of the states.
Specifically, each canonical-basis vector $\bm{e}_i$ represents the dynamical state 
$s_i$
(namely the $i$th state in the ordering)
because 
$\bm{e}_i$ as a probability vector represents the certainty of being in the state $s_i$. 

Each probability vector~$\bm{p}$ can be visualized in the dynamical graph as follows: each entry of~$\bm{p}$ represents the probability of the corresponding vertex of the dynamical graph. In other words, we are labeling the vertices of the dynamical graph with the entries of $\bm{p}$,
i.e., vertex~$s_i$ is labeled~$p_i$,
which we imagine as a continuous greyscale coloring of vertices in the graph.

The functions $\left\{\phi^n\right\}_{n\in\mathbb{N}}$ describing evolution in time can be
extended by linearity
\begin{align}\label{eq:linearized phi}
    \phi^n:&\,\mathbf{P}^M(\mathbb{R})\to\mathbf{P}^M(\mathbb{R})\nonumber\\
   :&\, \bm{e}_i\mapsto\bm{e}_j
\end{align}
such that $\phi^n\left(s_i\right)=s_j$, where $s_i$ is the state associated with the vector $\bm{e}_i$,
for every $n\in\mathbb{N}$. With a little abuse of notation, we also write the mapping in Eq.~\eqref{eq:linearized phi} as $\bm{e}_i \mapsto\bm{e}_{\phi\left(i\right)}$. 
As the generator of the dynamics~$\phi$ is a linear map, it can be
represented by a matrix
$\Phi$, called \emph{dynamical matrix}, once we fix the canonical basis in the domain and the codomain of $\phi$: its columns are
given by the action of $\phi$ on the canonical basis. Note that~$\Phi$ can be viewed as a column-stochastic matrix: every entry is non-negative and columns sum to 1 (there is only one 1 per column). In this spirit, $\Phi_{ij}$
is the probability for state~$s_j$
to transition to state~$s_i$.
According to graph theory, the matrix~$\Phi$
is nothing but the transpose of the adjacency matrix of the dynamical graph.

\subsection{Random Boolean networks}
\label{subsec:Random-Boolean-networks}

Random Boolean networks, introduced in 1969 by Kauffman~\citep{Kauffman1,Kauffman2},
were the first successful theoretical model to explain genetic regulation in cells.
They have been the privileged model to understand the expression of
genes for fifty years~\citep{Kauffman-review}, attracting the attention
of biologists and physicists working on complex systems as well~\citep{Gershenson-review,Drossel-review}.
A random Boolean network can be described as a discrete dynamical system; as such, it can be represented  as a dynamical graph.
\subsubsection{Introduction}
\label{subsubsec:booleannetwork}

We now discuss underlying concepts of random Boolean networks.
First we explain how the network represents which genes are expressed and which are not expressed at a given time.
Then we elaborate on how Boolean functions describe evolution from one configuration of gene expression to another.
Finally, we explain what is random per se in a random Boolean network.

Random Boolean networks are directed graphs, 
whose~$N$ vertices represent genes,
and directed edges represent the mutual influence of genes on the
expression of other genes.

Each vertex can be labeled black or white,
representing an expressed or an unexpressed gene, respectively.
We can represent the state of all
genes, i.e.\ of the whole network,
as a Boolean string
\begin{equation}
g=g_0g_1g_2\cdots g_{N-1}\in\{0,1\}^N
\end{equation}
(once we fix a total ordering of the vertices)
such that~$g_i$
represents expression or not of the $i$th gene, whether suppressed
($g_i=0$) or expressed ($g_i=1$). An example of such a network is presented in Fig.~\ref{fig:RBN}.

\begin{figure}
\includegraphics[width=\columnwidth]{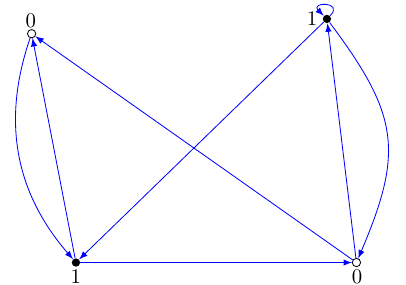}
\caption{%
This random Boolean network contains 4 genes (represented by the vertices of the graph). White vertices represent unexpressed genes; black vertices represent expressed ones. Directed edges represent which genes influence one another: the gene at the tip of an edge is influenced by the gene at the tail of the edge.}
\label{fig:RBN}
\end{figure}

The state of each gene~$g_i$ evolves to the next time step according to the states of the~$K_i$ parent genes in the previous time step,
i.e.\ graphically, those genes from which  
a directed edge points to the $i$th gene.
Mathematically, this random Boolean network describes
evolution of the state~$g$ of the whole network according to~$N$ Boolean functions
\begin{equation}
f_i:\left\{ 0,1\right\}^{K_i}\rightarrow\left\{ 0,1\right\}
\end{equation}
on the $i$th vertex with the output determining the black-or-white label for whether gene~$g_i$ is expressed or not
(black or white, respectively)
at the next time step.
All Boolean functions~$\{f_i\}$
 stay the same at every time step. 
In this setting, the expression
of a gene at each time step depends on which genes are expressed at the
previous step.

The vertices of the directed graph, representing genes, are fixed. Instead, randomness is only initial, and comes in three ways:
\begin{enumerate}
    \item randomness in which vertices are labeled black or white at the initial step;
    \item randomness in the choice of directed edges between the vertices;
    \item randomness in the choice of Boolean function associated with each vertex.
\end{enumerate}
Once these features are chosen, they stay the same for all time steps. In practice, the evolution can still be regarded as deterministic. For each of these three forms of randomness, various prior probability distributions have been proposed~\cite{Drossel-review}.  To represent the randomness in the choice of the initial state $g$ of the whole network, we resort to probability vectors, as discussed above. Here we have~$M=2^N$ states~$\{g\}$ of the whole network. Therefore, randomness in the initial state is described by a probability vector $\bm{p}\in \mathbf{P}^{2^N}(\mathbb{R})$. As the evolution can still be regarded as deterministic, for the scope of this article, we do not need a precise mathematical treatment of the other two types of randomness.

\subsubsection{Random Boolean networks as discrete dynamical systems}
\label{subsubsec:booleannetworkDDS}

Now we explain how random Boolean networks can be viewed as discrete dynamical systems. First, we show how to represent time evolution, and then how to construct the associated dynamical graph.

The evolution of a random Boolean network arising from the edges of the graph, and the choice of Boolean functions can be represented in a more high-level fashion. Take the set of Boolean strings $\{g\}$, which, as we saw above, represent the states of the network.
The action on the Boolean functions of the network makes a bit string
\begin{equation}
g=g_0g_1g_2\dots g_{N-1}
\end{equation}
evolve to a new bit string
\begin{equation}
g'=g'_0g'_1g'_2\dots g'_{N-1},
\end{equation}
where~$g'_i$ denotes the evolution of~$g_i$ at the next time step under the action of~$f_i$.
We can represent the evolution of $g$ into $g'$ in a single time step as a function $\phi:\left\{ g\right\} \rightarrow \left\{g\right\}$ on the set of bit strings of $N$ elements. Then the pair $\left(\{g\},\phi^n\right)$ is a discrete dynamical system.

Given that a random Boolean network can be described as a dynamical system, one can construct the associated dynamical graph, which is also known as the ``state space'' in the literature on random Boolean networks~\cite{Gershenson-review, Drossel-review}.
The rules to construct it are the same as presented for generic discrete dynamical systems: vertices are states of the network, and we use $\phi$ to draw directed edges.

\subsubsection{Network of genes and mutual interaction}
\label{subsubsec:GRNs}

Now we present the biological significance of random Boolean networks, explaining how they have been studied in the literature. Specifically, different phases have been identified as far as the dynamics in the ``thermodynamic limit'' are concerned.
In this setting, different attractors of the associated dynamical system are linked to different types of cells. Randomness arises because of the complexity in describing the noisy environment in which cells live.

Features of the discrete dynamical system associated with random Boolean network have a biological counterpart. For instance, it is argued that each attractor represents a different
\emph{cell type}~\citep{Kauffman1,Kauffman2,Kauffman-review}. This has also been verified experimentally~\cite{attractor-experiment}. 

Most of the analysis has been carried out in the so-called
``thermodynamic limit'', where one considers an ``infinite'' network,
mimicking the very large number of genes in the human DNA, and studying
the dynamics arising on it. In this scenario, a phase transition has
been observed, separating an ``ordered phase'', where the dynamics
are described by short and simple cycles, from a ``chaotic phase'',
where the dynamics are described by complex and long cycles. 

From the point of view of the adaptation of a cell to the environment,
the ordered phase is too static and rigid, and does not provide any
adaptive power to cells. On the other hand, the chaotic phase is too
unstable to guarantee a successful adaptation to the environment.
Therefore, Kauffman identified the ``golden phase'' to explain the
origin of the various cell types in the \emph{critical regime}, separating
the ordered phase from the chaotic one~\citep{Kauffman1,Kauffman2}.
According to the model of genetic regulation described by random Boolean
networks, the number of attractors in the critical regime follows
a power law: $N^{0.63}$, where $N$ is the number of genes~\citep{vanNimwegen2006,Kauffman-review}.
Empirical data show instead a power scaling given by $N^{0.88}$~\citep{Niklas2014,Kauffman-review}.
This discrepancy in the predictions has led some to question the actual
value of random Boolean networks as a model for genetic regulation
\citep{Drossel-review};
nevertheless they still remain a strong and
powerful model, and in any case, are interesting from the point of
view of the physics of complex systems.

In the study of random Boolean networks, there is the \emph{implicit} presence of an external environment, defined as everything that is outside the cell under consideration.
We assume that the environment is computationally hard to describe;
therefore, it is conveniently represented by the introduction of randomness in the form of probability vectors, as described above~\cite{Kauffman-environment,Gershenson-review,Drossel-review}. 

\subsection{The framework of resource theories}
\label{subsec:resource}
In this subsection we introduce resource theories, which are a powerful mathematical framework for dealing with the notion of ``resources'' in a rigorous fashion, developed in quantum information~\cite{Quantum-resource-1,Quantum-resource-2,Gour-review,Resource-channels-1,Resource-channels-2,Gour-Winter} and beyond~\cite{Resource-theories,Resoure-monotones,Chiribella-Scandolo-entanglement,Fong,Purity,ScandoloPhD,Nicole,Takagi-Regula,Selby2020compositional}.
 We start by explaining their relevance to quantum theory, and then we continue by formulating them in a theory-independent way, using process theories. This  allows us to extend their scope to discrete dynamical systems in the rest of this article.
\subsubsection{Motivation}
\label{subsubsec:motivation}

Resources theories have been remarkably successful in quantum information. providing new insights into the manipulation of quantum resources, such as quantum entanglement. The notion of a resource theory is based on the fact that, in several cases, we cannot implement all possible operations because of some constraints. The processes that respect the constraints are called ``free'' because they do not require any resources to occur.

Resource theories were developed to manage resources of quantum information optimally, in particular quantum entanglement~\cite{Review-entanglement,Plenio-review}. However, they have been extended to several other quantum information areas, such as quantum thermodynamics~\cite{delRio,Lostaglio-thermo}, the study of quantum coherence~\cite{Review-coherence} and quantum reference frames~\cite{Reference-frames-review,Gour-asymmetry1,Gour-asymmetry2,Marvian-asymmetry,Noether}, and quantum computing~\cite{Veitch_2014,Magic-states}. In these areas they have produced many important results, making a huge contribution to the development of new research areas.

In many cases, constraints restrict what processes can occur, in the sense that breaking such constraints is possible, but requires additional resources. For instance, in a thermodynamic setting a system is often immersed in an environment at a certain temperature. In this situation, the system reaches thermal equilibrium spontaneously, according to the ``minus first law'' of thermodynamics~\cite{BU01}.
This law restricts what thermodynamic processes can occur: they can only bring a state closer and closer to the thermal equilibrium with the environment. This can be argued to be the origin of the time-asymmetry in thermodynamics~\cite{BU01}. All the processes that go against the minus first law are a resource because they allow us to restore part of the time symmetry, and indeed they require some external work to be implemented. For this reason, whenever there is a restriction, the allowed processes are called ``free operations'', and those that enable us to overcome the restriction are called ``non-free''. Among processes, a particular type of them are those corresponding to the preparation or initialization of a system, which are identified with the states of the system itself. Again, free states correspond to preparations that can be done in the presence of the restriction, whereas non-free states require some resource to be implemented.

The goal of a resource theory is
to describe which transitions between the states of the theory are
possible in the presence of the restriction, namely under free operations.
This is known as the ``conversion problem''.

\subsubsection{Process theories}
\label{subsubsec:processtheories}

Now we introduce process theories~\citep{Coecke2011,Coecke2017picturing,Coecke-Kissinger}, which are a mathematical framework where the notion of process plays center stage. In these theories, one can describe systems as well as their evolution and interaction through processes. Their abstract formalisation is based on the notion of \emph{strict symmetric monoidal category}~\citep{baez_lauda_2011,Categories}. Process theories provide the necessary background to formulate resource theories outside the quantum setting.

In a process theory, systems are treated as labels
to identify different inputs/outputs. Systems can be composed to yield a composite system. At this stage, it is useful to introduce a particular type $I$, called the \emph{trivial system}, which represents the lack of a system. Clearly composing any
system~$a$ with $I$ yields~$A$ itself. Example of systems are chemical species, in chemistry, or Hilbert spaces arising in quantum physics.

As the name suggests, the core of a process theory are processes. Intuitively, a process is anything that happens between systems that has zero or more inputs and
zero or more outputs~\citep{Coecke-Kissinger}. The canonical example
is a function, which has one or more inputs, and has an output. However,
we can find more physical examples, such as chemical reactions, physical transformations.
Therefore, we represent a process as $f:A\rightarrow B$, where~$a$ is the input
type, and $B$ is the output type. When they share a type, two processes
$f:A\rightarrow B$ and $g:B\rightarrow C$, like functions, can be
composed sequentially to obtain a new process $g\circ f$. This means feeding the output of~$f$ into
the input of $g$. It is natural to require such a composition to
be associative, and to have an identity---the \emph{identity process}---which
corresponds to ``doing nothing''. This construction makes a process
theory a \emph{category }\citep{Categories}, with systems as objects
and processes as morphisms. When we have a composite type, $AB$ we can consider processes
that run ``independently'' or in parallel on~$a$ and $B$, i.e.\ two processes
$f:A\rightarrow C$ and $g:B\rightarrow D$ that occur ``at the same
time'' taking~$a$ and $B$ as input, respectively.
Clearly, these
two processes occurring together must yield a valid process on the
composite system $AB$. This new process, called the parallel composition of~$f$ and $g$, is denoted
by
\begin{equation}
f\otimes g:AB\rightarrow CD
\end{equation}
It is natural to assume some properties of the parallel composition that make it interact nicely with the sequential composition~\cite{Coecke-Kissinger}.

So far, we have not touched on the issue of the order of composition of systems, namely whether $AB=BA$. Intuitively,
we do not expect these two composite systems to be very different.
However, there in some cases, there is a conceptual difference. For example, 
consider a flask divided into two parts, containing two different
gases~$a$ and $B$, one in each part. $AB$ can stand for
``gas~$a$ is in part 1 and gas $B$ is in part 2''; therefore, $BA$
stands for ``gas $B$ is in part 1 and gas~$a$ is in part 2''.
Although these two settings are equivalent from the point of view
of the chemical species involved, we cannot say they are exactly the
same setting. Therefore, we can say that they are ``equivalent''.
To make this notion of equivalence mathematically precise, we introduce
the notion of \emph{swap}. A swap allows us to swap the order of two
systems in parallel composition.

The properties just seen make a process theory a \emph{strict symmetric
monoidal category}~\citep{Categories}. In this setting, we can view \emph{states} as a particular kind of a
process: a process with no input, corresponding to the preparation of a system.
Mathematically, they are processes with the trivial system as input; e.g.\ $s:I\rightarrow A$
will be a state of $A$. Therefore, in general, 
to understand what the states of a process theory are, we need to
identify what the trivial system is, i.e.\ the identity of system composition.

\subsubsection{Resource theories}
\label{subsubsec:resourcetheories}

From a formal point of view, a resource theory arises from a process theory when there is a restriction on the processes.
In this setting, the conversion problem becomes relevant, which means determining whether, given two states, we can convert one into the other. The answer can be given with the aid of particular functions of the two states, called conversion witnesses. 
Finally, we present the example of the resource theory of asymmetry, which provides the mathematical inspiration for the approach in this article.

A resource theory is a process theory augmented by partitioning processes into ``free'' and ``non-free'' processes such that composing free processes and swapping systems are free as well as the identity process.
These requirements make a resource theory equivalent to specifying a strict
symmetric monoidal subcategory of the process theory. The morphisms
in the subcategory are exactly the free operations of the resource
theory. In particular, the states in this subcategory are the free states of the resource theory. 

We can set up a hierarchy
among the states of a theory, whereby a state is more valuable than
another if,
from the former,
it is possible to reach a larger set of
states using only free processes~\citep{Resource-theories,Gour-review}. More formally, we can set up a partial preorder, known as the \emph{resource preorder},
defined as $s' \precsim s$,
if there exists a free process~$f$ such that $s'=fs$, where the product denotes sequential composition of the processes~$f$ and~$s$.
In this setting, we want to map the set of states with its partial preorder into a better known or simpler partially ordered set $\left(X,\leq\right)$~\cite{Fong}. To this end, we need to introduce a function~$M$ from the set of states to a partially ordered set~$X$ such that $M\left(s'\right)\leq M\left(s\right)$ whenever $s' \precsim s$.
Such a function is called a \emph{monotone}. Normally, the condition $M\left(s'\right)\leq M\left(s\right)$ is only a necessary condition for the conversion of~$s$ into~$s'$. However, in some special cases, combining possibly more than one monotone (even involving different partially ordered sets, as we show in~\S\ref{subsec:conversion} even involving the divisibility partial order, denoted by $\mid$, can give us necessary and sufficient conditions for the convertibility under free operations. In this case, we can solve the conversion problem completely, and such a set of monotones is called a \emph{complete set}.

A particularly successful example of a resource theory in quantum information is the resource theory of 
of asymmetry~\citep{Gour-asymmetry1,Gour-asymmetry2,Marvian-asymmetry}. A symmetry of a system corresponds to it being unchanged under certain actions, such as time reversal, spacial reflections. A quantum system is associated with a Hilbert space, and its symmetries are represented by unitary or anti-unitary operators on it. In the resource theory of asymmetry, for every quantum system, there is a group representing its symmetry, whose action is given by unitary operators up to a global phase. Quantum states are described by trace-class positive semi-definite operators with unit trace. Quantum processes are completely positive and trace non-increasing linear maps on the set of self-adjoint operators on a Hilbert space. As seen above, free operations arise from some restriction. In this case it is a \emph{covariance condition}: $\mathcal{E}$ is free if $\mathcal{U}_g \circ \mathcal{E} =  \mathcal{E} \circ \mathcal{U}_g$, for every unitary operation $\mathcal{U}_g$ associated with the symmetry (in the context of group representation theory, such maps are known as equivariant maps). Thus, acting with the symmetry before and after the action of $\mathcal{E}$ is the same. This fact can be illustrated with a commutative diagram.
\begin{equation}
\begin{CD}\mathcal{B}\left(\mathcal{H}\right)@>\mathcal{E}>>\mathcal{B}\left(\mathcal{H}\right)\\
@V\mathcal{U}_{g}VV@VV\mathcal{U}_{g}V\\
\mathcal{B}\left(\mathcal{H}\right)@>\mathcal{E}>>\mathcal{B}\left(\mathcal{H}\right)
\end{CD}
\end{equation}
The resource theory of asymmetry was used to give an informational reconstruction of Noether's theorem~\cite{Noether}.

\section{Approach}
\label{sec:approach}

In this section,
we discuss how we develop a mathematical theory of evolutionarity,
which is based on the model of discrete dynamical systems as such systems represent all systems evolving in discrete systems.
We begin by showing that discrete dynamical systems are special cases of process theories.
Then we show how the covariance condition,
discussed in~\S\ref{subsubsec:resourcetheories},
is applied to discrete dynamical systems, and it gives rise to a full-fledged resource theory.
Finally, we explain how we treat influences,
usually called perturbations,
on evolving systems.
\subsection{Formulating discrete dynamical systems}
\label{subsec:dynamicalgraphproctheory}

Now we present two new labels to introduce on states of a discrete dynamical systems (or, equivalently, on vertices of the dynamical graph) that are useful for stating our results. We also provide a formulation of a discrete dynamical system as a process theory both in the absence and in the presence of randomness.

\subsubsection{New labels}
\label{subsubsec:new labels}

Here we introduce two new labels for states of a dynamical system. The first, called transient progeny, looks at their evolution in the future, specifying how fast they reach an attractor, viz.\ dynamics that occurs in the long run. The second, called ancestry, looks at their evolution from the past, identifying how many steps were necessary to evolve to that state.

It is possible to quantify how transient a state~$s$ is by considering its \emph{transient progeny}, (or progeny for short)
the minimum natural number $d$ such that $\phi^{d}\left(s\right)$
is in a cycle. This also quantifies the number of successors of~$s$ that are still transient, including~$s$ itself. The bigger the progeny, the more transient the state is. This is
the length of the shortest path from the vertex associated with that state to a vertex in a cycle in the dynamical graph.
From this point of view, the states of a cycle can be viewed as transient
states of progeny 0 because 0 time steps are necessary to evolve
them to a state in a cycle.

We can also consider a ``backwards distance'' of a state from its
farthest predecessor in the dynamical graph. More formally, we define the \emph{ancestry}
$a$ of a vertex as the maximum length of a path leading to that vertex.
Equivalently, the ancestry of a state $s$, denoted as
$a\left(s\right)$, is the maximum natural number for which there
exist a state~$t$ such that $\phi^{a\left(s\right)}\left(t\right)=s$.
For states~$s$ in an attractor of length~$\ell$, since we have $\phi^{n\ell}\left(s\right)=s$
for all $n\in\mathbb{N}$, the ancestry is formally infinite, so we
set $a\left(s\right)=+\infty$. Instead for transient states, their ancestry
will be a finite natural number.

\subsubsection{Non-random case}
\label{subsubsec:deterministic initial}

In~\S\ref{subsec:DDS}, we explained that,
in a discrete dynamical system $\left(S,\phi\right)$ with no randomness present at any step of its evolution, evolution is described by powers of $\phi$. It is natural to describe this situation with a process theory based on sets and functions~\cite{Coecke-Kissinger}.

As we are dealing with discrete dynamical systems, we construct a process theory where systems are in fact dynamical systems, i.e.\ pairs $\left(S,\phi\right)$, where~$S$
is a set and $\phi:S\rightarrow S$ is the generator of the dynamics on~$s$. Processes between dynamical systems $\left(S_1,\phi_1\right)$ and $\left(S_2,\phi_2\right)$ are arbitrary functions between $S_1$ and $S_2$. In particular, when $S_1=S_2$, we have what we called \emph{influences}.

The composition of dynamical systems is by Cartesian product:
if $\left(S_1,\phi_1\right)$ and $\left(S_2,\phi_2\right)$ are two dynamical systems, their composition is $\left(S_1 \times S_2,\phi_1  \times\phi_2\right)$,
where
\begin{equation}
\phi_1  \times\phi_2:S_1 \times S_2\rightarrow S_1 \times S_2:
\left(s_1,s_2\right)
\mapsto\left(\phi_1\left(s_1\right),\phi_2\left(s_2\right)\right)
\end{equation}
for every $s_1\in S_1$, and every $s_2\in S_2$.
The sequential composition of processes is simply the composition of functions.
Instead, the parallel composition of processes \begin{equation}
f:S_1\rightarrow S_2,\,
f':S'_1\rightarrow S'_2
\end{equation}
is the function
\begin{equation}
f\times f':S_1 \times S'_1\rightarrow S_2 \times S'_2
\end{equation}
such that
\begin{equation}
\left(s_1,s'_1\right)\mapsto \left(f\left(s_1\right),f'\left(s'_1\right)\right)
\end{equation}
for every $s_1\in S_1$ and $s'_1\in S'_1$.

In this setting, it is clear that the trivial dynamical system $\left(I,\phi\right)$, which does not affect other systems when composed with them, is obtained by taking $I$ to be a singleton and $\phi$ to be the identity $\mathds1$.
With this in mind, the process-theoretic states of $\left(S,\phi\right)$,
viz.\ processes from $I$ to~$S$ can be identified with the elements of $S$~\cite{Categories,Coecke2017picturing}, i.e.\ with what we defined as the states of the dynamical system in~\S\ref{subsec:DDS}.
In the following, these states are called ``deterministic states'' to differentiate them from the states permitting randomness.

\subsubsection{Random case}
\label{subsubsec:random initial}

In~\S\ref{subsubsec:introducingrandomness}, we explained that in a discrete dynamical system $\left(S,\phi\right)$ in the presence of randomness, we represent randomness with probability vectors and evolution is described by powers of a linearization of $\phi$. This conceptually distinct situation, can be described naturally based on the process theory of stochastic maps~\cite{Coecke-Kissinger}.

The potential presence of randomness in  a discrete dynamical system $\left(S,\phi\right)$ at any time step forces us to use a different process theory from the case of deterministic initial conditions. Now, if~$S$  has~$M$ elements, systems are pairs $\left(\mathfrak{S},\phi\right)$, where $\mathfrak{S}$ is the simplex of probability vectors in  $\mathbb{R}^M$, namely probability distributions over~$s$. Processes between systems $\left(\mathfrak{S}_1,\phi_1\right)$ and $\left(\mathfrak{S}_2,\phi_2\right)$ are stochastic maps between $\mathfrak{S}_1$ and $\mathfrak{S}_2$.
\begin{definition}\label{def:stochastic influence}
We call \emph{stochastic influence} any stochastic map
$f:\mathfrak{S}\to\mathfrak{S}$.
\end{definition}
Now a stochastic map
\begin{equation}
f:\mathfrak{S}_1\rightarrow\mathfrak{S}_2
\end{equation}
can be represented with
a stochastic matrix once we fix the canonical bases in the domain $\mathfrak{S}_1$ and the codomain $\mathfrak{S}_2$ of~$f$:
\begin{equation}
\label{eq:stochastic}
F=\begin{pmatrix}
p\left(1|1\right) & \cdots& p\left(1|M_1\right)\\
\vdots & \ddots & \vdots\\
p\left(M_2|1\right) & \cdots& p\left(M_2|M_1\right)
\end{pmatrix}.
\end{equation}
The $\left(i,j\right)$ entry is the probability $p\left(i|j\right)$
of jumping to the state of underlying set $S_2$ labeled by $i$ from the state of the underlying set $S_1$ labeled by $j$. This probability is 1 if and only if there is a deterministic transition. The $j$th column
of~$f$ is the evolution of the deterministic state labeled by $j$ under the
covariant stochastic map $f$. 
If the influence is deterministic,
as defined in~\S\ref{subsec:DDS},
we refer this influence now as `deterministic influence'. Note that a deterministic influence only has zeros and ones as entries in its representation as a stochastic matrix.

The composition of a system $\left(\mathfrak{S}_1,\phi_1\right)$ with a system $\left(\mathfrak{S}_2,\phi_2\right)$ is the system $\left(\mathfrak{S}_1\otimes\mathfrak{S}_2,\phi_1\otimes\phi_2\right)$, where $\mathfrak{S}_1\otimes\mathfrak{S}_2$ is the simplex of probability vectors in the tensor product of the real vector spaces associated with $\mathfrak{S}_1$ and $\mathfrak{S}_2$.
The sequential composition of processes is simply their composition as linear functions.
Instead, the parallel composition of processes
\begin{equation}
f:\mathfrak{S}_1\rightarrow\mathfrak{S}_2,\,
f':\mathfrak{S}'_1\rightarrow\mathfrak{S}'_2
\end{equation}
is the tensor-product function
\begin{equation}
f\otimes f':\mathfrak{S}_1\otimes\mathfrak{S}'_1
\to\mathfrak{S}_2\otimes\mathfrak{S}'_2.
\end{equation}
In every process theory, there is a trivial system corresponding to the system that does not affect other in composition.

In the presence of randomness, this system is  $\left(\mathfrak{S}_0,\mathds1\right)$, where $\mathfrak{S}_0$ is the simplex of probability vectors in $\mathbb{R}$, i.e.\ the set containing only the number~1.
With this in mind, the process-theoretic states of a generic dynamical system in the presence of randomness $\left(\mathfrak{S},\phi\right)$ can be identified with the elements of the simplex $\mathfrak{S}$ itself, i.e.\ with probability vectors. In the following, these states are called ``stochastic states''.

\subsection{Imposing covariance}
\label{subsec:imposingcovariance}

In this subsection we introduce the mathematical tools to describe the influence of the external environment on a discrete dynamical system. We focus on \emph{covariant} influences, which are external influences that somehow respect the evolution of a dynamical system. We also show that restricting to covariant influences gives rise to a well-defined resource theory.

\subsubsection{Introducing covariance}
\label{subsubsec:introducingcovariance}

Now we give the formal definition of covariant maps in the process theory of discrete dynamical systems, both in the presence of deterministic and random initial conditions. We also present their foremost mathematical properties.

The definition of covariant maps can be given in a unified way in both process theories we defined above, namely for deterministic and random initial conditions. Recall that, in the case of no randomness allowed, systems are pairs $\left(S,\phi\right)$, where~$s$ is a set, and $\phi$ is the generator of the dynamics. On the other hand, in the case of allowed randomness, systems are pairs $\left(\mathfrak{S},\phi\right)$, where $\mathfrak{S}$ is the set of probability vectors, and now $\phi$ is the linear version of the generator of the dynamics, as described in~\S\ref{subsubsec:introducingrandomness}. We cover both situations by representing systems as $\left(A,\phi\right)$, where~$a$ is either~$S$ or $\mathfrak{S}$.
\begin{definition}
\label{def:covariance}
Let $\left(A_1,\phi_1\right)$ and $\left(A_2,\phi_2\right)$ be dynamical systems. A  map $f:A_1\rightarrow A_2$ is called \emph{covariant} if $f\circ\phi_1=\phi_2\circ f$.
\end{definition}

A covariant map in the absence of randomness is a function between sets, and, in the presence of randomness, is a stochastic map.
In particular, in the case of a stochastic influence $f$, the associated matrix~$F$ is square and the covariance condition takes a simple matrix form, namely $\left[F,\Phi\right]=0$
for stochastic matrix~$\Phi$ associated with $\phi$, the generator of dynamics.

As~$\Phi$ describes a deterministic dynamic, there will be only one non-zero entry per column, and such an entry will be equal to 1.
The commutation condition $\left[F,\Phi\right]=0$, signaling a covariant influence, will add further constraints to the entries of $F$. Such constraints in general depend on~$\Phi$.

Note that, in particular, all powers of $\phi$, namely, \begin{equation}
    \phi^n:A \rightarrow A,\;
    n\in\mathbb{N},
\end{equation}
are covariant, as
\begin{equation}
    \phi^n\circ\phi=\phi\circ\phi^n=\phi^{n+1}.
\end{equation}
\begin{proposition}
\label{prop:n-covariant}
Let $\left(A_1,\phi_1\right)$ and $\left(A_2,\phi_2\right)$ be dynamical systems. A map $f:A_1\to A_2$
is covariant if and only if (iff)
\begin{equation}
\label{eq:covariant}
f\circ\phi_1^n
    =\phi_2^n\circ f,\;
     n\in\mathbb{N}
\end{equation}
\end{proposition}
\begin{proof}
Sufficiency is straightforward: it is enough to take $n=1$ to have
the thesis.

To prove necessity, we proceed by induction. If $n=0$ or $n=1$
the property is obviously true. Let us assume it is true for a generic
$n$, and let us prove is holds also for $n+1$. We have
\begin{align}
f\circ\phi_1^{n+1}=&f\circ\phi_1\circ\phi_1^n=\phi_2\circ f\circ\phi_1^n\nonumber\\
=&\phi_2\circ\phi_2^n\circ f=\phi_2^{n+1}\circ f,
\end{align}
where,
in the second equality,
we used the covariance of~$f$, and in
the third equality the induction hypothesis. This proves the proposition.
\end{proof}
\subsubsection{Covariance yields a resource theory}
\label{subsubsec:yields}

Now we show that restricting ourselves to covariant maps yields strict symmetric monoidal subcategories of the process theories we introduced in to describe discrete dynamical systems. Thus, we think of covariant maps as free operations of a new resource theory.

\begin{proposition}
Covariant maps form a strict symmetric monoidal subcategory.
\end{proposition}
\begin{proof}
The identity, being part
of the dynamics, is obviously covariant. The sequential composition
of two covariant maps is still covariant.
Indeed, if $\left(A_1,\phi_1\right)$, $\left(A_2,\phi_2\right)$, and $\left(A_3,\phi_3\right)$ are dynamical systems,
and if 
\begin{equation}
f:A_1\rightarrow A_2,\,
g:A_2\rightarrow A_3
\end{equation}
are both covariant, then
\begin{equation}
g\circ f:A_1\rightarrow A_3
\end{equation}
satisfies
\begin{equation}
    g\circ f\circ\phi_1=g\circ\phi_2\circ f=\phi_3\circ g\circ f,
\end{equation}
where we use the covariance of~$f$ in the first equality and the
covariance of~$g$ in the second equality. Similarly, the parallel
composition of covariant maps is still a covariant map. This is because
the generator of dynamics on the composition of two discrete dynamical systems is
given by the parallel composition
of the generators of the respective dynamics. Hence,
if $\left(A_1,\phi_1\right)$, $\left(A_2,\phi_2\right)$, $\left(A_3,\phi_3\right)$, and $\left(A_4,\phi_4\right)$ are dynamical systems, and if
\begin{equation}
f:A_1\rightarrow A_2,\,
g:A_3 \rightarrow A_4,
\end{equation}
then
\begin{align}
\left(f\otimes g\right)&\circ\left(\phi_1\otimes\phi_3\right)
=\left(f\circ\phi_1\right)\otimes\left(g\circ\phi_3\right)\nonumber\\
=&\left(\phi_2\circ f\right)\otimes\left(\phi_4\circ g\right)=\left(\phi_2\otimes\phi_4\right)\circ\left(f\otimes g\right).
\end{align}
The swap between two discrete dynamical systems $\left(A_1,\phi_1\right)$, $\left(A_2,\phi_2\right)$ is a covariant operation:
\begin{equation}
\mathtt{SWAP}\circ\left(\phi_1\otimes\phi_2\right)=\left(\phi_2\otimes\phi_1\right)\circ\mathtt{SWAP}.
\end{equation}
as
it immediately follows from the property of the \texttt{SWAP} itself.
\end{proof}
\noindent
This proposition shows that the definition of covariant operations is well posed,
and it gives rise to a full-fledged resource theory.
Therefore, we
can rightfully consider them as free operations.

\subsubsection{Covariant influence}
\label{subsec:detvsran}

Now we explain the notion of influence on a dynamical system, introduced in~\S\ref{subsec:DDS} and Definition~\ref{def:stochastic influence},
as the effective representation of the action of something external to the dynamical system itself. External influence happens in several situations, given the fact that truly isolated systems do not really exist.
We also motivate the introduction of covariant influence,
which arises due to the dynamical time scale of evolution being shorter than the time scale for changes due to influences external to the dynamical system.

Discrete dynamical systems are often studied as isolated systems, where one is interested in how the system evolves subject to its own dynamics. However, in Nature truly isolated system do not exist, and there is always something surrounding a dynamical system---an environment---that influences it. For this reason, despite the fact that a dynamical system and its environment can be jointly treated as a larger, isolated, dynamical system, it is still interesting to focus purely on the dynamical system and describe the presence of the environment through some effective action that disrupts the evolution of the dynamical system. This action is a function that maps states of the dynamical system, whether deterministic or stochastic, to other states. As such, we model this action by incorporating influence, whether deterministic or stochastic.

Among all possible influences, we focus on covariant ones, namely those obeying Definition~\ref{def:covariance}. These influences, despite changing the evolution of the system, work in such a way that it is compatible with the dynamics.
Indeed, the covariance condition can be stated as follows: if we evolve the dynamical system for $n$ time steps and then we disturb it, it is the same as first disturbing the system with the same influence and then letting the system evolve for $n$ time steps.
The idea behind this is that a system has time to adapt to the external influence, perhaps because such an influence occurs on a longer time scale. In particular, as we show in~\S\ref{sec:results}, a covariant influence cannot undo the evolution of state of a dynamical system under its dynamic.

\section{Results}
\label{sec:results}

In this section we present the results of this article. First we identify free states in the resource theory of covariant influences, which allows to identify two different types of resources into play whether randomness is forbidden or not. In the former case, we have the resource theory of evolutionarity, and in the latter case, the resource theory is of non-attractorness. Finally, we turn to the conversion problem, i.e.\ determining when one state can be converted into another by means of a covariant influence.

\subsection{Free states}
\label{subsec:freestates}

In this subsection we partition the states of a dynamical system into free and non-free ones according to the resource theory of covariant influences. Free states are  very different whether randomness is forbidden or not. For this reason, we treat these two situations separately.

\subsubsection{Deterministic case~}

Now we examine the situation in which randomness is completely forbidden, both in the initial conditions
(i.e., the deterministic state)
and in the changes induced by a covariant influence.
This corresponds to a dynamical system whose complexity is low enough to allow us to use a fully deterministic description.

In the context of resource theories, free deterministic states of a system~$s$ are covariant maps $f:I\rightarrow S$, where $I$ is a singleton. We show that deterministic free states can be identified with fixed points of a discrete dynamical system~\eqref{eq:fixed point}.
\begin{proposition}\label{prop:free deterministic}
There is a bijection between free deterministic states of a discrete dynamical system $\left(S,\phi\right)$ and its fixed points.
\end{proposition}
\begin{proof}
Consider a free deterministic state $f:I\to S$; this sends the only element
of $I$, which we call $1$ by convention, to a state~$s$ of~$s$. The dynamic on the trivial dynamical system $I$ is given by the identity; therefore, the
condition for covariance becomes $f=\phi\circ f$.
Applying~$f$ to
1, and recalling $f\left(1\right)=s$, we have $s=\phi\left(s\right)$.
This means that~$s$ is
a fixed point.

Conversely, given a fixed point~$s$ of the evolution of $\left(S,\phi\right)$,
i.e.\ $\phi\left(s\right)=s$, we can construct the unique covariant function
$f:I\rightarrow S$ such that $f\left(1\right)=s$. The covariance
of~$f$ is guaranteed by the fact that~$s$ is a fixed point.
\end{proof}

Not all discrete dynamical systems have fixed points; this depends on the
dynamics. By Proposition~\ref{prop:free deterministic}, those systems with no fixed points, in the lack of randomness, have no free states.

\subsubsection{Stochastic case~}
\label{subsubsec:nonattractornesscase}

Now we analyze the situation of a dynamical system where randomness is permitted. Thus, we can have stochastic states at any time step and also the influence can be stochastic. This situation models dynamical systems where their behavior is too complicated to be described feasibly in purely deterministic terms, due to the presence of an influence that acts on the dynamical system as a stochastic process.
We see that, for stochastic dynamical systems, the landscape of free states becomes richer.

As seen in~\S\ref{subsec:DDS},
states of a stochastic dynamical system can be represented as probability vectors~$\bm{p}$.
In this setting, free stochastic states are defined as maps
\begin{equation}
f:\mathfrak{S}_0\rightarrow \mathfrak{S},
\end{equation}
obeying the covariance condition of Def.~\ref{def:covariance},
where
now $f\left(1\right)=\bm{p}$ is still equivalent to $\phi\left(\bm{p}\right)=\bm{p}$.
By a similar argument to that for deterministic free states, the covariance condition on stochastic maps entails the probability vector must be invariant under the generator of the dynamics.

Below we prove equivalence between the existence of a bijection between free stochastic states and probability distributions over the attractors of the discrete dynamical system with the requirements that every deterministic state in an attractor sharing equal probability with all other states in the same attractor.
Before formalizing this proposition,
we introduce pertinent notation,
building on the ability to partition deterministic states into basins of attraction.
For~$\bm{b}_i$ referring to block corresponding to the $i^\text{th}$ basin of attraction,
we write
\begin{equation}\label{eq:block decomposition}
\bm{p}=\bm{b}_1\oplus\cdots\oplus\bm{b}_k.
\end{equation}
with~$k$ basins of attraction and~$\bm p$ the probability vector.

Let us focus on the first block as,
for other blocks,
the analysis is completely analogous. Note that, for transient deterministic
states~$s$, by taking large enough~$m$, we can always ensure that
$\phi^{m}\left(s\right)$ is in an attractor.
For example, it is
enough to take~$M$ as the maximum of the transient progenies of transient states.
In $\bm{b}_1$ there are $\ell_1$ states corresponding to
the deterministic states in the attractor (of length $\ell_1$)
and other~$m_1$ transient deterministic states, i.e.\ states that
are not in the attractor, but eventually goes into it.
\begin{definition}
A stochastic state is called uniform if two conditions hold:
\begin{enumerate}
    \item all probabilities associated with transient deterministic states are 0;
    \item deterministic states in the same attractor have the same probability.
\end{enumerate}
\end{definition}
\noindent
In simpler terms, every attractor has a uniform distribution over
its deterministic states (with uniform probabilities $\nicefrac1{\ell}$, $\ell$
being the length of the attractor), weighted by a probability distribution
over the different attractors.

Below and henceforth we consider integers modulo some number, say~$m$. This means picking a convenient representative for each remainder class modulo~$M$ so that expressions are well defined. 
Here we consider two representatives for each remainder class. The one we pick depends on the concrete setting where we consider integers modulo 1: in some cases it may mean picking a representative from 1 to~$m$, in others from 0 to~$m-1$.

\begin{proposition}
\label{prop:free random states}
There is a bijection
between free stochastic states of a discrete dynamical system and uniform probability vectors.
\end{proposition}
\begin{proof}
We start by proving that if a probability vector $\bm{p}$,
representing a stochastic state,
is free, that is, invariant under $\phi$, then it must be uniform. Following Eq.~\eqref{eq:block decomposition}, we arrange
the deterministic states $s_i$ corresponding to the first block~$\bm{b}_1$ so that
the first $\ell_1$ entries of $\bm{b}_1$ are the deterministic states of its attractor, in such a way that $\phi\left(s_i\right)=s_{i+1}$, where $i+1$ is intended to be modulo $\ell_1$:
thus, we pick a representative in the remainder classes modulo $\ell_1$ between 1 and $\ell_1$. 
By Prop.~\ref{prop:n-covariant},
the request $\phi\left(\bm{p}\right)=\bm{p}$ is equivalent
to $\phi^n\left(\bm{p}\right)=\bm{p}$ for every $n\in\mathbb{N}$.
In particular, we require  covariance for $n=m$, where~$M$ is the minimum exponent for which $\phi^m\left(s\right)$ is in the attractor, for every deterministic state~$s$ in the first basin of attraction. We obtain
\begin{equation}
\bm{b}_1=\left(\begin{array}{c}
p_1\\
\vdots\\
p_{\ell_1}\\
\hline p_{\ell_1+1}\\
\vdots\\
p_{\ell_1+m_1}
\end{array}\right)=\left(\begin{array}{c}
p'_{1-m}\\
\vdots\\
p'_{\ell_1-m}\\
\hline 0\\
\vdots\\
0
\end{array}\right)=\phi^{m}\left(\bm{b}_1\right),\label{eq:probability covariant}
\end{equation}
where $p'_i$ is either $p_i$ or $p_i$ plus a probability
contribution coming from the transient states, and the subscripts
of the entries of the vector in right-hand side are  intended to be
modulo $\ell_1$, in the same way as discussed above.
Equation~\eqref{eq:probability covariant}
immediately tells us that the probabilities associated with transient
states must all vanish.
Now,
if we require the invariance under $\phi$
of such a $\bm{b}_1$ with 0 probabilities for transient states,
we obtain
\begin{equation}
\bm{b}_1=\begin{pmatrix}
p_1\\
\vdots\\
p_{\ell_1}\\
\hline 0\\
\vdots
\end{pmatrix}=\left(\begin{array}{c}
p_{\ell_1}\\
p_1\\
\vdots\\
p_{\ell_1-1}\\
\hline 0\\
\vdots
\end{array}\right)=\phi\left(\bm{b}_1\right).
\end{equation}
This yields the condition $p_1=\cdots=p_{\ell_1}$.

To conclude the proof, we repeat the same argument for each of the
$k$ blocks. Therefore, the probabilities associated with the deterministic
states of an attractor must be equal for the states in that
attractor (but possibly different between different attractors).

Conversely, by a similar argument as above, one can easily check that uniform states are invariant under $\phi$.
\end{proof}

The situation of a discrete dynamical system differs a lot whether randomness is allowed or not. Indeed, when randomness is forbidden, free states are only fixed points, whereas when randomness is permitted, free states are associated with attractors of any length. In particular, it means that free stochastic states always exist because any discrete dynamical system has attractors.

\subsection{The conversion problem: evolutionarity}
\label{subsec:conversion}

Now we analyze the conversion problem for discrete dynamical systems where randomness is forbidden. In this case, we call the resource ``evolutionarity'' as all valuable states (i.e.\ non-free) evolve in time because they are not fixed points. We are able to give a necessary and sufficient condition for convertibility under covariant operations, and we show that covariant influences shorten the length of attractors.

Denoting the length of a deterministic state $s$  by $\ell$, its transient progeny by $d$, and its ancestry by $a\left(s\right)$, we have the complete characterization of the state transitions within a single dynamical system that can occur due to deterministic influences, presented in the following theorem.
\begin{theorem}
\label{thm:necessary and sufficient} Let~$s$ and~$s'$ be two deterministic
states of a discrete dynamical system $\left(S,\phi\right)$. Then there exists
a covariant  influence converting~$s$ into~$s'$ iff
\begin{align}
\label{eq:condcovinfl}
&d'\leq d,\\
&\ell'\mid\ell,\label{eq:condcovdiv}\\
&a\left(\phi^n\left(s'\right)\right)\geq a\left(\phi^n\left(s\right)\right)\label{eq:ancestry condition}
\end{align}
for $n=0,\ldots,d'-1$.
\end{theorem}
\noindent
In the following, we provide the proof of this result, dividing it into different pieces.

Theorem~\ref{thm:necessary and sufficient} gives necessary and sufficient conditions for the conversion of deterministic states of a discrete dynamical system under covariant influences. We can phrase these conditions in the language of monotones (see \S\ref{subsec:resource}). Indeed, we can view the transient progeny $d$ of a deterministic state as a function from the set of deterministic states to the set $\mathbb{N}$ of natural numbers, seen as the (totally) ordered set $\left(\mathbb{N},\leq\right)$. The length, instead, can be obtained from a function from the set of deterministic states and $\mathbb{N}$, this time seen as the partially ordered set $\left(\mathbb{N},\mid\right)$, where $\mid$ denotes the divisibility partial order. Finally, the ancestry, is a function from the set of states to $\mathbb{N}\cup \left\{+\infty\right\}$, totally ordered by the $\geq$ order. According to Theorem~\ref{thm:necessary and sufficient}, these functions are monotones in our resource theory. However, not only can they be used to formulate necessary conditions for the conversion of states, but also sufficient~\eqref{eq:condcovinfl}--\eqref{eq:ancestry condition}.
Thus, they form a complete set of monotones.

\subsubsection{Necessary conditions~}

As a first step to solve the conversion problem, we examine necessary conditions for the conversion of one deterministic state into another. These conditions involves labels of deterministic states: the lengths, the transient progenies and the ancestries of the states involved.

We saw that, when randomness is forbidden, there is a bijection between free states and fixed points of the discrete dynamical system. Thus, all non-free states, i.e.\ valuable states, are states that evolve in time. This makes us identify the relevant resource as ``evolutionarity'', which is the property that a deterministic state evolves.

Let us consider two deterministic states~$s$ and $s'$ of the same discrete dynamical system $\left(S,\phi\right)$ (the case of states of two distinct dynamical systems is treated in Appendix~\ref{app:different dynamical}).
We want to determine when covariant influences can transform~$s$ into~$s'$. Here we formulate necessary conditions.
 The first involves the transient progeny and the length~$\ell$ of a state.
\begin{lemma}
\label{lem:transient necessary}In a discrete dynamical system without randomness,
states of transient progeny $d$  and of length $\ell$
are mapped to states of transient progeny $d'\leq d$ and of length $\ell'$, such that $\ell'\mid\ell$. 
\end{lemma}
\begin{proof}
Let~$s$ be a deterministic state of transient progeny $d$ and
of length $\ell$ in a discrete dynamical system $\left(S,\phi\right)$.
Then 
\begin{equation}
\phi^{\ell+d}\left(s\right)=\phi^{\ell}\left(\phi^{d}\left(s\right)\right)=\phi^{d}\left(s\right).
\end{equation}
Let $f:S\rightarrow S$
be covariant;
then
\begin{align}
\phi^{d}\left(f\left(s\right)\right)
=&f\left(\phi^{d}\left(s\right)\right)=
f\left(\phi^{\ell+d}\left(s\right)\right)
    \nonumber\\
=&\phi^{\ell+d}\left(f\left(s\right)\right).
\end{align}
This means that, after $\ell$ iterations, we come back to $\phi^{d}\left(f\left(s\right)\right)$,
which implies that $\phi^{d}\left(f\left(s\right)\right)$
lies in an attractor of length $\ell'$ that divides $\ell$.
This also means that $f\left(s\right)$ will be a state of transient progeny
 $d'\leq d$.
\end{proof}
\noindent Thus, transient states become ``less transient''
(in particular they become attractor states if $d'=0$).
\begin{remark}
If $d=0$, the state~$s$ is an attractor state, and Lemma~\ref{lem:transient necessary} tells us how any state evolves,
including attractor states,
under covariant operations. Specifically, the constraint  $0\leq d'\leq d=0$, implies 
$d'=0$.
If~$s$ is an attractor state,
then~$s'$ must be an attractor state too, whose length divides the length of the original attractor. In other words, attractor states are sent to attractor states.
\end{remark}

We can also say something about the ancestry of states: it increases
under free operations. Clearly, this statement is meaningful when
at least one of the states is transient.
\begin{lemma}
\label{lem:ancestry necessary}For every state~$s$ of a deterministic
dynamical system $\left(S,\phi\right)$, the ancestry of the successors of~$s$ does not decrease
under free operations.
Formally, if  $f:S\rightarrow S$ is covariant, then
\begin{equation}\label{eq:ancestry}
a\left(\phi^n\left(s\right)\right)\leq a\left(\phi^n\circ f\left(s\right)\right),
\end{equation}
for every $n\in\mathbb{N}$.
\end{lemma}
\begin{proof}
If~$s$ is an attractor state, all its successors are attractor 
states too. This means, that for every $n\in\mathbb{N}$, $a\left(\phi^n\left(s\right)\right)=+\infty$.
By Lemma~\ref{lem:transient necessary}, $f\left(s\right)$ is an
attractor state, so, again its successors are attractor states.
Therefore, $a\left(\phi^n\circ f\left(s\right)\right)=+\infty$
for every $n\in\mathbb{N}$, satisfying the condition $a\left(\phi^n\left(s\right)\right)\leq a\left(\phi^n\circ f\left(s\right)\right)$
for every natural number $n\in \mathbb{N}$.

Suppose now~$s$ is a transient state, and let $d>0$ be its transient progeny. This means that, after $d$ steps,
the successors of~$s$ become attractor states. 
By Lemma~\ref{lem:transient necessary}, $f\left(s\right)$ is a state with transient progeny $d'\leq d$, which means that, after $d'\leq d$
steps the successors of $f\left(s\right)$ become attractor states.
Therefore, for $n\geq d'$, we have $a\left(\phi^n\circ f\left(s\right)\right)=+\infty$,
which clearly satisfies inequality~\eqref{eq:ancestry}.
Instead, for every $n=0,\ldots,d'-1$, there exists a state $s_n^*$ such that $\phi^n\left(s\right)=\phi^{a\left(\phi^n\left(s\right)\right)}\left(s_n^*\right)$.
By covariance of $f$,
\begin{align}
\phi^n\circ f\left(s\right)
=&f\left(\phi^n\left(s\right)\right)=f\left(\phi^{a\left(\phi^n\left(s\right)\right)}\left(s_n^*\right)\right)\nonumber\\
=&\phi^{a\left(\phi^n\left(s\right)\right)}\left(f\left(s_n^*\right)\right).
\end{align}
This shows that $f\left(s_n^*\right)$ is a predecessor of $\phi^n\circ f\left(s\right)$,
and that there are $a\left(\phi^n\left(s\right)\right)$ steps
from $f\left(s_n^*\right)$ to~$\phi^n\circ f\left(s\right)$.
Thus, Eq.~\eqref{eq:ancestry} holds
even for $n=0,\ldots,d'-1$. In conclusion, we proved that, for any
state $s$, Eq.~\eqref{eq:ancestry} holds
for every $n\in\mathbb{N}$.
\end{proof}
\noindent In words, a (transient) state becomes farther and
farther from its farthest predecessor.

In particular, in the two lemmas above we can take~$f$ to be the generator $\phi$ of the dynamics.
When $f=\phi$, $d'\leq d$, $\ell=\ell'$, and the ancestry
never decreases for all the successors of a state. Transitions from one state to another due to dynamics occur only inside a given basin of attraction. However, the statements of the two lemmas suggest that, when~$f$ is not the dynamic, it may be possible to jump 
between different basins of attraction, provided the constraints
are met, e.g.\ from an attractor of length 4 to an attractor of length 2. However, at the moment we do not know if this is the case because the two lemmas are only \emph{necessary} conditions, but a priori not sufficient. 
\subsubsection{Sufficient conditions~}

The conditions expressed in Lemmas~\ref{lem:transient necessary} and~\ref{lem:ancestry necessary} are necessary, and they tell us how deterministic states change under covariant operations. Now we prove that these conditions are sufficient too: if they are satisfied, we can conclude that a covariant map exists that sends a deterministic state~$s$ to~$s'$.

\begin{lemma}\label{lem:sufficient1}
Let~$s$ and~$s'$ be two deterministic
states of discrete dynamical system $\left(S,\phi\right)$, with~$s'$ an attractor state,
and $\ell,\ell'$ their respective  lengths.
If 
$\ell'\mid\ell$, there exists
a covariant operation converting~$s$ into~$s'$.
\end{lemma}
\begin{proof}
To prove sufficiency, we must construct a covariant function mapping
$s$ into~$s'$.
To this end, we need to specify its action on all
deterministic states of $\left(S,\phi\right)$.
Now, for the states
$\{x\}$ \emph{not} in the basin of attraction of $s$, construct $f\left(x\right):=x$.
Such an~$f$ is covariant on these states by construction. Now we need to map the
basin of attraction of~$s$ to the basin of attraction of~$s'$ in a covariant way. The
starting point is to fix $f\left(s\right)=s'$, and then we reconstruct
the action of~$f$ on the other states of the basin of attraction
of~$s$.

Observe that~$s$ has $d+\ell$ successors
$s_n:=\phi^n\left(s\right)$
(including itself $s_0=s$), for $n=0,\ldots,d+\ell-1$, where $d$ is the transient progeny of~$s$. For these states~$\{s_n\}$, construct $f\left(s_n\right):=\phi^n\left(s'\right)$,
which yields in particular $f\left(s\right)=s'$.
Now, for $n=0,\ldots,d+\ell-2$,
one has
\begin{align}
\phi\circ&f\left(s_n\right)
=\phi\circ\phi^n\left(s'\right)=\phi^{n+1}\left(s'\right)\nonumber\\
=&f\left(s_{n+1}\right)=f\left(\phi^{n+1}\left(s\right)\right)=f\circ\phi\left(s_n\right),
\end{align}
and, for $n=d+\ell-1$,
\begin{equation}
\phi\circ f\left(s_{d+\ell-1}\right)=\phi\circ\phi^{d+\ell-1}\left(s'\right)=\phi^{d+\ell}\left(s'\right)
=\phi^{d}\left(s'\right)
\end{equation}
because~$s'$ is in an attractor of length $\ell'\mid\ell$.
On the other hand, 
\begin{align}
f\circ&\phi\left(s_{d+\ell-1}\right)=f\circ\phi\circ\phi^{d+\ell-1}\left(s\right)\nonumber\\
=&f\circ\phi^{d+\ell}\left(s\right)=f\circ\phi^{d}\left(s\right)=f\left(s_d\right)
=\phi^{d}\left(s'\right),
\end{align}
where we have used the fact that~$s$ is a state of transient progeny $d$
and length~$\ell$. We conclude that $\phi\circ f\left(s_{d+\ell-1}\right)=f\circ\phi\left(s_{d+\ell-1}\right)$.
In conclusion,~$f$ is covariant on every successor of~$s$.

Now we need to define~$f$ on all the other states in the same basin of
attraction as~$s$.
Note that every state in the same basin of attraction
as~$s$ is a predecessor of $s_d:=\phi^{d}\left(s\right)$ because
$s_d$ is in the attractor of length~$\ell$.
Let $\mathcal{S}$ be the set
of successors of $s$ (cf.\ also Fig.~\ref{fig:graphsets}).
\begin{equation}
\label{eq:setofsuccessors}
\mathcal{S}
=\left\{s_n:=\phi^n\left(s\right);
n=0,\ldots,d+\ell-1
\right\}.
\end{equation}
For the states~$t$ in the complement of $\mathcal{S}$,
if it is non-empty,
define $\delta\left(t; s_d\right)$ to be
the minimum number of steps to reach $s_d$ from $t$.
Note that these states are all transient because all the attractor states are
among the successors of~$s$. On the states in the (non-empty) complement
of $\mathcal{S}$, set~$f$ to be $f\left(t\right):=\phi^{d-\delta\left(t; s_d\right)}\left(s'\right)$, 
where the exponent is regarded to be modulo $\ell'$, where this time we pick a representative in the remainder classes modulo $\ell'$ that is between 0 and $\ell'-1$.

Let us check
that this definition makes~$f$ covariant. We have $\phi\circ f\left(t\right)=\phi^{\delta\left(t; s_d\right)+1}\left(s'\right)$.
Instead, to assess $f\circ\phi\left(t\right)$,
we must distinguish
two cases.
If $\phi\left(t\right)\notin \mathcal{S}$, we have
\begin{align}
f\circ\phi\left(t\right)
    =&\phi^{d-\delta\left(\phi\left(t\right); s_d\right)}\left(s'\right)
    \nonumber\\
    =&\phi^{d-\delta\left(t; s_d\right)+1}\left(s'\right),
\end{align}
where we have used the fact that $\delta\left(\phi\left(t\right); s_d\right)=\delta\left(t; s_d\right)-1$
for transient states $t$. Indeed, if $\phi^{k}\left(t\right)=s_d$,
then $\phi^{k-1}\left(\phi\left(t\right)\right)=s$ (and for transient
states $k>0$). Instead, if $\phi\left(t\right)\in \mathcal{S}$, then $\phi\left(t\right)=\phi^n\left(s\right)$,
for some $n\in\left\{ 0,\ldots,d+\ell-1\right\} $. Therefore, by
the above definition of $f$, $f\circ\phi\left(t\right)=\phi^n\left(s'\right)$.
Now, if $n\leq d$, we have $\delta\left(\phi\left(t\right); s_d\right)=d-n$,
whence $\delta\left(t; s_d\right)=d-n+1$. In this case
\begin{equation}
f\circ\phi\left(t\right)=\phi^n\left(s'\right)=\phi^{d-\delta\left(t; s_d\right)+1}\left(s'\right).
\end{equation}
Since~$s'$ is in an attractor of period $\ell'$, the exponent must
be regarded as modulo $\ell'$, as discussed above. On the other hand, if $n\in\left\{ d+1,\ldots,d+\ell-1\right\} $,
we have $\delta\left(\phi\left(t\right); s_d\right)=d-n$ (modulo
$\ell$, where for every remainder class modulo
$\ell$ we pick a representative between 0 and $\ell$), whence $\delta\left(t; s_d\right)=d-n+1$ (modulo~$\ell$, as before).
Again,
\begin{equation}
f\circ\phi\left(t\right)=\phi^n\left(s'\right)=\phi^{d-\delta\left(t; s_d\right)+1}\left(s'\right),
\end{equation}
where the expression at exponent is modulo~$\ell$, as above. However, since~$s'$ is in a
cycle of length $\ell'$, the exponent is actually modulo $\ell'$, in the sense above, because of the length $\ell'$ of the cycle. Given that $\ell'\mid \ell$, whether we take the exponent modulo $\ell$ or $\ell'$, the state $\phi^{d-\delta\left(t; s_d\right)+1}\left(s'\right)$ is the same. In conclusion,~$f$ is covariant on all states
in the complement of $\mathcal{S}$ too. This proves that there is a covariant influence
$f$ mapping~$s$ to~$s'$.
\end{proof}

With reference to Lemma~\ref{lem:sufficient1}, if $d$ and $d'$ denote the transient progeny of~$s$ and~$s'$ respectively, we have $d'=0$ because~$s'$ is an attractor state, and
$a\left(\phi^n\left(s'\right)\right)=+\infty$ for every $n\in\mathbb{N}$,
because all successors of~$s'$ are attractor states too.
Hence, note that this
situation can be viewed as a special case of the more general conditions $d'\leq d$ and $a\left(\phi^n\left(s'\right)\right)\geq a\left(\phi^n\left(s\right)\right)$
for $n=0,\ldots,d'-1$, which are examined for the case of transient states in the following lemma.

\begin{example}
\label{example:freeattractormapping4to2}
The example of how to construct
a free operation mapping an attractor of length 4 into an attractor
of length 2 is shown in Fig.~\ref{fig:freeopconstruction}.
\hfill
$\blacksquare$
\begin{figure}
\includegraphics[width=\columnwidth]{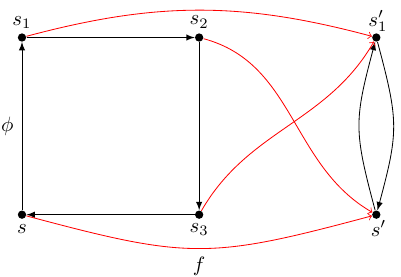}
\caption{%
How a covariant influence maps an attractor~$\{s,s_1,s_2,s_3\}$ of length 4 into an attractor~$\{s',s'_1\}$ of length 2.
Black arrows represent the action of~$\phi$,
and red arrows represent the action of the covariant influence~$f$. Here we have $f\left(s\right) = s'$. Then, $f\left(s_1\right)=\phi\circ f\left(s\right)=\phi\left(s'\right)=s_1'$. Similarly, $f\left(s_2\right)=\phi\circ f\left(s_1\right)=\phi\left(s_1'\right)=s'$ and $f\left(s_3\right)=\phi\circ f\left(s_2\right)=\phi\left(s'\right)=s_1'$.}
\label{fig:freeopconstruction}
\end{figure}
\end{example}

\begin{lemma}
\label{lem:sufficient2}Let~$s$ and~$s'$ be two deterministic
states of a deterministic dynamical systems $\left(S,\phi\right)$, with~$s'$ transient. If $d'\leq d$,
$\ell'\mid\ell$ and $a\left(\phi^n\left(s'\right)\right)\geq a\left(\phi^n\left(s\right)\right)$
for $n=0,\ldots,d'-1$, then there exists
a covariant influence converting~$s$ into~$s'$ 
\end{lemma}
\begin{proof}
 Again, we have to construct a covariant function mapping
$s$ into $s'$, specifying its action on all
deterministic states of $\left(S,\phi\right)$. As in the proof of Lemma~\ref{lem:sufficient1}, for the states
$x$ \emph{not} in the basin of attraction of $s$, define $f\left(x\right):=x$. To construct the action of~$f$ on the states in the same basin of attraction as $s$, we proceed as in the proof of Lemma~\ref{lem:sufficient1}.

 The conditions $d'\leq d$ and $a\left(\phi^n\left(s'\right)\right)\geq a\left(\phi^n\left(s\right)\right)$
for $n=0,\ldots,d'-1$ imply that~$s$ is transient too. The situation is now 
more complicated than in the proof of Lemma~\ref{lem:sufficient1}, due to the presence of states in
the basin of attraction of~$s$ that are neither successors nor predecessors
of~$s$.

 Consider the set $\mathcal{S}$ of the $d+\ell$ successors of $s$
as in Eq.~\eqref{eq:setofsuccessors}.
As done
in the proof of Lemma~\ref{lem:sufficient1}, construct $f\left(s_n\right):=\phi^n\left(s'\right)$,
which yields $f\left(s\right)=s'$. Now, for $n=0,\ldots,d+\ell-2$,
one immediately proves the covariance of $f$, as seen in the proof of Lemma~\ref{lem:sufficient1}. Instead,
for $n=d+\ell-1$, one has
\begin{align}
\phi\circ f\left(s_{d+\ell-1}\right)
=&\phi\circ\phi^{d+\ell-1}\left(s'\right)
\nonumber\\
=&\phi^{\ell}\left(\phi^{d}\left(s'\right)\right)=\phi^{d}\left(s'\right),
\end{align}
because~$s'$ has transient progeny $d'\leq d$ and length
$\ell'\mid\ell$, so we know that $\phi^{d}\left(s'\right)$ is in
the attractor in the same basin of attraction as~$s'$. Now, using the fact that~$s$ is a transient state
of transient progeny $d$ and length~$\ell$, as above, we have $f\circ\phi\left(s_{d+\ell-1}\right)=\phi^{d}\left(s'\right)$, because $\phi\left(s_{d+\ell-1}\right)=s_d$
In conclusion,~$f$ is covariant on every successor of~$s$.\\

 Now consider the set of predecessors
of~$s$, namely,
\begin{equation}
\mathcal{P}_{0}=\left\{ t:\phi^{k}\left(t\right)=s,\textrm{ for some }k\in\mathbb{N}\right\},
\end{equation}
which contains~$s$ itself (obtained for $k=0$), as represented in Fig.~\ref{fig:graphsets}.

To respect the dynamics, these states must be mapped to predecessors
of~$s'$. We are only interested
in the states in $\mathcal{P}_{0}$ that are different from~$s$ (provided they exist). If they do not exist, we move to the next step, i.e.\ the predecessors of the successors of~$s$. Instead, if $\mathcal{P}_{0}\backslash\left\{ s\right\} $
is non-empty, for the states $t\in \mathcal{P}_{0}\backslash\left\{ s\right\} $,
let $\delta\left(t; s\right)$ be the number of steps to go from
$t$ to~$s$. We know that the maximum value $\delta\left(t; s\right)$
is exactly $a\left(s\right)$. Now, by hypothesis $a\left(s'\right)\geq a\left(s\right)$,
so it is possible to find a predecessor~$s'^*$ of~$s'$ such that
\begin{equation}
\label{eq:Ds'*s}
    \delta\left(s'^*; s\right)=a\left(s\right).
\end{equation}
Now, for $t\in \mathcal{P}_{0}\backslash\left\{ s\right\} $, construct
\begin{equation}\label{eq:predecessors of s}
    f\left(t\right):=\phi^{a\left(s\right)-\delta\left(t; s\right)}\left(s'^*\right).
\end{equation}
Note that the exponent is always non-negative because $\delta\left(t; s\right)\leq a\left(s\right)$.

 Let us show that Definition~\eqref{eq:predecessors of s} makes~$f$ covariant on $\mathcal{P}_{0}\backslash\left\{ s\right\} $ when such a set is non-empty.
We have
\begin{equation}\label{eq:predecessors of s1}
\phi\circ f\left(t\right)=\phi^{a\left(s\right)-\delta\left(t; s\right)+1}\left(s'^*\right).
\end{equation}
On the other hand, to calculate $f\circ\phi\left(t\right)$, we have
to distinguish two cases.
The first case is when $\phi\left(t\right)=s$, which means $\delta\left(t; s\right)=1$. In this case, we have $f\circ\phi\left(t\right)=s'$. Now, Eq.~\eqref{eq:predecessors of s1}
reads
\begin{align}
\phi\circ f\left(t\right)
=&\phi^{a\left(s\right)-\delta\left(t; s\right)+1}\left(s'^*\right)
=\phi^{a\left(s\right)-1+1}\left(s'^*\right)
\nonumber\\
=&\phi^{a\left(s\right)}\left(s'^*\right)
=s'.
\end{align}

 The second case is when $\phi\left(t\right)\neq s$.
Observe that $\phi\left(t\right)$
cannot be a state $s_n$ for $n>0$, i.e.\ a successor of~$s$. Indeed, if this were the case,
we would have $\phi\left(t\right)=\phi^n\left(s\right)$, for some
$n\geq1$. Since~$t$ is in $\mathcal{P}_{0}\backslash\left\{ s\right\} $,
we have $\phi^{k}\left(t\right)=s$, for some $k\geq1$. Combining
these two statements we obtain 
\begin{equation}
s=\phi^{k-1}\circ\phi\left(t\right)=\phi^{k-1}\circ\phi^n\left(s\right)=\phi^{k-1+n}\left(s\right).
\end{equation}
Since~$s$ is transient, this is possible iff $k-1+n=0$.
Since $n>0$ and $k\geq1$, this condition can never be satisfied.
Additionally, $\phi\left(t\right)$ can neither be one of the (possible)
states $w$ in the same basin of attraction as~$s$ that are neither successors
nor predecessors of~$s$. Indeed, if $\phi\left(t\right)$ were such
a $w$, we would have
\begin{equation}
s=\phi^{k-1}\circ\phi\left(t\right)=\phi^{k-1}\left(w\right),
\end{equation}
which would mean that $w$ is a predecessor of $s$, which is against
the hypothesis. Therefore, if $\phi\left(t\right)\neq s$, the only possibility is that $\phi\left(t\right)$ is a predecessor of~$s$. In this case, 
\begin{align}
f\circ\phi\left(t\right)=\phi^{a\left(s\right)-\delta\left(\phi\left(t\right); s\right)}\left(s'^*\right)
\nonumber\\
=\phi^{a\left(s\right)-\delta\left(t; s\right)+1}\left(s'^*\right),
\end{align}
All of this proves that the constructed~$f$ is covariant on $\mathcal{P}_{0}\backslash\left\{ s\right\} $.

 In the same basin of attraction as~$s$ there can be some states $w$
that are neither successors nor predecessors of~$s$. In any case, they
are predecessors of some of the successors of~$s$. Then, for every
successor $s_n$, with $n=1,\ldots,d+\ell-1$,
consider the set of its predecessors
\begin{equation}
\mathcal{P}_n=\left\{ t:\phi^{k}\left(t\right)=s_n,\textrm{ for some }k\in\mathbb{N}\right\} .
\end{equation}
Note that, for every $n=1,\ldots,d+\ell-1$, $\mathcal{P}_m\subseteq \mathcal{P}_n$
for $m<n$. Indeed, consider a state $t\in \mathcal{P}_m$, with $m<n$.
Then $\phi^{k}\left(t\right)=\phi^{m}\left(s\right)$, for some $k\in\mathbb{N}$.
However, as $m<n$, we can consider
\begin{align}
\phi^{n-m+k}\left(t\right)
=&\phi^{n-m}\circ\phi^{k}\left(t\right)
=\phi^{n-m}\circ\phi^{m}\left(s\right)
\nonumber\\
=&\phi^n\left(s\right).
\end{align}
Therefore, $t\in \mathcal{P}_n$, showing the inclusion $\mathcal{P}_m\subseteq \mathcal{P}_n$.
To avoid defining~$f$ redundantly on the sets $P_n$, for $n=1,\ldots,d+\ell-1$,
we have to prune the sets $\mathcal{P}_n$, getting rid of the successors
of~$s$ in $\mathcal{S}$~(\ref{eq:setofsuccessors}),
of the states in $\mathcal{P}_{0}$ (for which~$f$ has already
been defined), and taking care of the inclusions $\mathcal{P}_m\subseteq \mathcal{P}_n$
for $m<n$. To this end, we define the sets
\begin{equation}
\widetilde{\mathcal{P}}_n:=\mathcal{P}_n\backslash\left(\mathcal{P}_{n-1}\cup \mathcal{S}\right),
\end{equation}
for $n=1,\ldots,d+\ell-1$.
Now, we need to consider only the $\widetilde{\mathcal{P}}_n$'s
that are non-empty (see Fig~\ref{fig:graphsets}).

\begin{figure}
\includegraphics{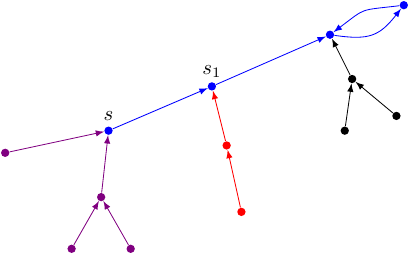}
\caption{\label{fig:graphsets}
A state $s$ in a dynamical graph, with $d=2$ and $\ell=2$. The set $\mathcal{S}$ of the $2+2=4$ successors of $s$ is represented in blue. The set $\mathcal{P}_0$ of the predecessors of $s$ is the purple set, along with $s$ itself. The set $\widetilde{P}_1$ is depicted in red, and represents a side chain that ends up in the main blue chain. The set $\mathcal{P}_1$, the set of predecessors of $s_1$, is made up of the purple states, of the red states and $s$ and $s_1$. By considering all sets $\widetilde{P}_n$ for $n=1,\dots,d+\ell-1$, in the proof of Lemma~\ref{lem:sufficient2} we take care of all collateral chains ending up in $\mathcal{S}$.}

\end{figure}

 Note that the maximum $n$ for which $\widetilde{\mathcal{P}}_n$
can be non-empty is $n=d$. Indeed, since, for $n\geq d$, $s_n$
is in the attractor, $\mathcal{P}_n$ contains all the states of the basin, and so does $\mathcal{P}_{n-1}$ if $n>d$,
so $\widetilde{\mathcal{P}}_n=\varnothing$ for every $n>d$. We distinguish
two cases: when $1\leq n<d'$, and when $d'\leq n\leq d$.

 In the first case, if $1\leq n<d'$,
for the states $w\in\widetilde{P}_n$, we know that $\delta\left(w; s_n\right)\leq a\left(\phi^n\left(s\right)\right)$.
By hypothesis, we know that $a\left(\phi^n\left(s'\right)\right)\geq a\left(\phi^n\left(s\right)\right)$,
so it is possible to find a predecessor $s_n'^*$ of $\phi^n\left(s'\right)$
such that
\begin{equation}
\label{eq:Dsn&phins'}
\delta\left(s_n'^*;\phi^n\left(s'\right)\right)=a\left(s_n\right).
\end{equation}
Therefore, if $w\in\widetilde{P}_n$ (with $n<d'$), we construct $f$
as
\begin{equation}
    f\left(w\right):=\phi^{a\left(s_n\right)-\delta\left(w; s_n\right)}\left(s_n'^*\right).
\end{equation}
Again, this function is well-defined because $\delta\left(w; s_n\right)\leq a\left(s_n\right)$.
Let us show that it is covariant.
To this end, we have \begin{equation}\label{eq: phi circ f}
    \phi\circ f\left(w\right)=\phi^{a\left(s_n\right)-\delta\left(w; s_n\right)+1}\left(s_n'^*\right).
\end{equation}
To assess $f\circ\phi\left(w\right)$, we need to distinguish two cases. The first is when $\phi\left(w\right)\in\widetilde{P}_n$. In this case,
\begin{align}\label{eq:f circ phi}
f\circ\phi\left(w\right)&=\phi^{a\left(s_n\right)-\delta\left(\phi\left(w\right); s_n\right)}\left(s_n'^*\right) \nonumber\\
&=\phi^{a\left(s_n\right)-\delta\left(w; s_n\right)+1}\left(s_n'^*\right).
\end{align}

 The other case is when $\phi\left(w\right)\notin\widetilde{P}_n$. We claim that 
the only possibility is $\phi\left(w\right)=s_n$. Indeed,
by a similar argument as above, $\phi\left(w\right)$ cannot be a
state  $s_m$ for $m>n$. To see why, suppose by contradiction
that $\phi\left(w\right)=\phi^{m}\left(s\right)$ for some $m>n$.
Since $w\in\widetilde{P}_n$, there exists $k\geq1$ such that $\phi^{k}\left(w\right)=\phi^n\left(s\right)$.
Then, combining these two properties
\begin{align}
\phi^n\left(s\right)&=\phi^{k-1}\circ\phi\left(w\right)=\phi^{k-1}\circ\phi^{m}\left(s\right)\nonumber\\
&=\phi^{k-1+m}\left(s\right).\label{eq:transient}
\end{align}
Here we assume $n<d'\leq d$;
therefore, $\phi^n\left(s\right)$
is still a transient state. Then Eq.~\eqref{eq:transient} is satisfied
iff $n=k-1+m$. This is equivalent to solving $\left(k-1\right)+\left(m-n\right)=0$.
Since $m>n$ and $k\geq1$, this expression can never vanish, so eq.~\eqref{eq:transient}
cannot be satisfied. Additionally, $\phi\left(w\right)$ cannot be in
$\mathcal{P}_{n-1}$, for otherwise there would exist $k'\geq0$ such that $\phi^{k'+1}\left(w\right)=s_{n-1}$,
contradicting the assumption that $w\notin \mathcal{P}_{n-1}$. Finally, $\phi\left(w\right)$
cannot be a (possible) state $w'$ that is neither a predecessor nor
a successor of $s_n$. Indeed, if this were the case, we would have
$\phi\left(w\right)=w',$which combined with $\phi^{k}\left(w\right)=s_n$, for some $k\geq 1$,
yields
\begin{equation}
s_n=\phi^{k-1}\circ\phi\left(w\right)=\phi^{k-1}\left(w'\right),
\end{equation}
which is against the hypothesis that $w'$ is not a predecessor of
$\phi^n\left(s\right)$. Therefore, we are only left with the possibility
$\phi\left(w\right)=s_n$.

 Now, if $\phi\left(w\right)=s_n$, we have $f\circ\phi\left(w\right)=\phi^n\left(s'\right)$. On the other hand, as $\phi\left(w\right)=s_n$,  we have $\delta\left(w; s_n\right)=1$. Thus, substituting into Eq.~\eqref{eq: phi circ f}
yields
\begin{equation}
\phi\circ f\left(w\right)=\phi^{a\left(s_n\right)-1+1}\left(s_n'^*\right)=\phi^n\left(s'\right).
\end{equation}
Here we have used Eq.~\eqref{eq:Dsn&phins'}, which states that $a\left(s_n\right)$ is the number of steps from $s_n'^*$ to $\phi^n\left(s'\right)$.
This shows covariance of~$f$ on $\widetilde{\mathcal{P}}_n$, for $n=1,\ldots,d'-1$.

 To conclude the proof, consider the case for which some $\widetilde{\mathcal{P}}_n$'s
are non-empty for $n=d',\ldots,d$.
For $w\in\widetilde{\mathcal{P}}_n$,
construct
\begin{equation}
f\left(w\right):=\phi^{n-\delta\left(w; s_n\right)}\left(s'\right),
\end{equation}
where the exponent is 
modulo $\ell'$, and we pick a representative of the remainder classes modulo $\ell'$ from $0$ to $\ell'-1$. Now we show that this definition makes~$f$ covariant
on $\widetilde{\mathcal{P}}_n$.
Clearly,
\begin{equation}
\phi\circ f\left(w\right)=\phi^{n-\delta\left(w; s_n\right)+1}\left(s'\right).
\end{equation}
To assess $f\circ\phi\left(w\right)$, as done above, we need to distinguish two cases. The first case is $\phi\left(w\right)\in\widetilde{\mathcal{P}}_n$,
in which case
\begin{align}\label{eq:f circ phi s'}
f\circ\phi\left(w\right)&=\phi^{n-\delta\left(\phi\left(w\right); s_n\right)}\left(s'\right)\nonumber\\
&=\phi^{n-\delta\left(w; s_n\right)+1}\left(s'\right).
\end{align}
This case shows that~$f$ is covariant when $\phi\left(w\right)\in\widetilde{\mathcal{P}}_n$, with $n=d',\dots,d$ now we proceed to examine the second case, which is when $\phi\left(w\right)\notin\widetilde{\mathcal{P}}_n$.

 To address the second case, i.e.\ $\phi\left(w\right)\notin\widetilde{\mathcal{P}}_n$, for $n=d',\dots,d$, we further distinguish two cases. If $n=d',\ldots,d-1$, $s_n=\phi^n\left(s\right)$ (i.e.\ we exclude $n=d$), then $s_n=\phi^n\left(s)\right)$ is a transient
state;
thus, as seen above, if $\phi\left(w\right)\notin\widetilde{\mathcal{P}}_n$,
then the only possibility is $\phi\left(w\right)=s_n$.
In this case, $f\circ\phi\left(w\right)=\phi^n\left(s'\right)$.
The fact  $\phi\left(w\right)=s_n$ implies that $\delta\left(w; s_n\right)=1$.
Then $f\circ\phi\left(w\right)$ (cf.\ Eq.~\eqref{eq:f circ phi s'}) takes the form
\begin{align}
f\circ\phi\left(w\right)&=\phi^{n-\delta\left(w; s_n\right)+1}\left(s'\right)\nonumber\\
&=\phi^{n-1+1}\left(s'\right)\nonumber\\
&=\phi^n\left(s'\right).
\end{align}
This concludes the proof of covariance when  $\phi\left(w\right)\notin\widetilde{\mathcal{P}}_n$ and $n=d',\ldots,d-1$.

 Now we tackle the final case, namely $\phi\left(w\right)\notin\widetilde{\mathcal{P}}_n$ and $n=d$. In this case, for $n=d$, $s_d$ is in the attractor, so $\phi\left(w\right)$
this time can be any of the states in the attractor.
Indeed Eq.~\eqref{eq:transient}, for $n=d$
is equivalent to $d\equiv k-1+m\mod\ell$, which has solutions
even for $m>d$ and $k\geq1$ (it is enough to take an $m>d$ such
that $m\equiv d-k+1\mod\ell$).
Therefore, we conclude that $\phi\left(w\right)=s_m$,
for some $m\in\left\{ d,\ldots,d+\ell-1\right\} $, which is any of the states in the attractor. In this case $f\circ\phi\left(w\right)=\phi^{m}\left(s'\right)$.
Now, $\delta\left(\phi\left(w\right); s_d\right)=d-m$ (modulo~$\ell$, where for every remainder class modulo
$\ell$ we pick a representative between 0 and $\ell$), which means $\delta\left(w; s_d\right)=d-m+1$
(modulo~$\ell$ as before).
Therefore,
\begin{equation}
f\circ\phi\left(w\right)=\phi^{m}\left(s'\right)=\phi^{d-\delta\left(w; s_d\right)+1}\left(s'\right),
\end{equation}
where the expression in the exponent is modulo~$\ell$, as above.  However, since~$s'$ is in a
cycle of length $\ell'$, the exponent is actually modulo $\ell'$, in the sense above, because of the length $\ell'$ of the cycle. Given that $\ell'\mid \ell$, whether we take the exponent modulo $\ell$ or $\ell'$, the state $\phi^{d-\delta\left(t; s_d\right)+1}\left(s'\right)$ is the same.
\end{proof}
\noindent
With Lemmas~\ref{lem:sufficient1} and~\ref{lem:sufficient2} we finally obtain Theorem~\ref{thm:necessary and sufficient}.

\subsection{The conversion problem: non-attractorness}
\label{subsec:nonattractorness}

 Now we analyze the conversion problem for discrete dynamical systems where randomness is allowed.
In this case, we call the resource ``non-attractorness'' as all valuable states (i.e.\ non-free) are non-uniform probability vectors;
i.e.\ they are are not associated with uniform probability distributions supported on attractor deterministic states.  We show that randomness acts as an activator of transitions between deterministic states: with stochastic covariant influences it is possible to jump from one deterministic state to another even if their length do not satisfy the constraints of Theorem~\ref{thm:necessary and sufficient}.

Given a stochastic dynamical system $\left(\mathfrak{S},\phi\right)$ and two stochastic states~$\bm{p}$ and $\bm{p}'$, we want to establish whether there exists a stochastic matrix~$F$ commuting with the matrix~$\Phi$ associated with $\phi$,  such that $\bm{p}'=F\bm{p}$. From a computational point of view, the solution to this problem can be achieved with a set of constrained linear equations. The fact that all constraints are expressed by linear equations makes the problem tractable. In fact, we can show that this is a linear decision problem (see Appendix~\ref{app:linear programming}).

We can still try to see if we can reduce the complexity of determining the solution to the conversion problem by reducing the number of linear conditions to impose. To achieve this, we can study if some constraints on the entries of the matrix~$F$ to be determined follow directly from general properties of commutation of~$F$ with some~$\Phi$, even without specifying what~$\Phi$ is actually like. For this reason, now we study how the covariance condition affects the entries of the most general square stochastic matrix $F$, which, according to Eq.~\eqref{eq:stochastic}, are the jumping probabilities between deterministic states.
For example, we can determine which entries of~$F$ are forced to be equal or to vanish
by the covariance condition.

In the stochastic case we are mainly interested in the conversion problem between probability vectors. However, the structure of matrices of covariant influences can also provide a new angle on transitions between deterministic states. Indeed, the entries of such matrices are transition probabilities between deterministic states. In deterministic dynamical systems, influences map a deterministic state~$s$ into one deterministic state~$s'$. In stochastic dynamical systems, in general influences will allow multiple transitions from a state $s$, each of which is weighted by a certain probability. In particular, we say that a transition from a deterministic~$s$ to a deterministic state $s'$ is allowed if the corresponding transition probability $p\left(s'\middle|s\right)$ is non-zero.
Therefore, in the light of Theorem~\ref{thm:necessary and sufficient}, we are interested in seeing if the presence of randomness activates some types of transitions between deterministic states that would otherwise be forbidden by the conditions of Theorem~\ref{thm:necessary and sufficient}. To answer this question, we need once more to understand which transition probabilities $p\left(i|j\right)$ are compatible with the covariance condition expressed in its most general form.
\begin{example}\label{ex:activated transitions}
 Consider a discrete dynamical system $\left(S,\phi\right)$ that features two cycles: one of length 2 and one of length~1 (i.e.\ a fixed point). Let us determine the constraints on a generic stochastic matrix acting on the probability vectors of that discrete dynamical system, arising from the commutation condition with the dynamical matrix~$\Phi$. We call $s_1$ and $s_2$ the two deterministic states in the cycle of length 2, and $s_3$ the fixed point. With this convention,  $p\left(i\middle|j\right)$ ($i,j\in\left\{1,2,3\right\}$) is the transition probability from the deterministic state $s_j$ to the deterministic state $s_i$.
The dynamical matrix is 
\begin{equation}
\Phi=\begin{pmatrix}
0 & 1 & 0\\
1 & 0 & 0\\
0 & 0 & 1
\end{pmatrix},
\end{equation}
where the first column tells us that $s_1$ gets mapped to $s_2$, the second column that $s_2$ gets mapped to $s_1$ (in agreement with the fact that we have a cycle of length 2), and the third column tells us that $s_3$ remains~$s_3$ because it is a fixed point.
Imposing the covariance condition $\left[F,\Phi\right]=0$ with a generic stochastic matrix $F$ yields new constraints on the entries of $F$:
\begin{align}
        &\begin{pmatrix}
    p\left(1\middle|1\right) & p\left(1\middle|2\right) & p\left(1\middle|3\right)\\
    p\left(2\middle|1\right) & p\left(2\middle|2\right) & p\left(2\middle|3\right)\\
    p\left(3\middle|1\right) & p\left(3\middle|2\right) & p\left(3\middle|3\right)
    \end{pmatrix}\begin{pmatrix}
0 & 1 & 0\\
1 & 0 & 0\\
0 & 0 & 1
\end{pmatrix}\nonumber\\&=\begin{pmatrix}
0 & 1 & 0\\
1 & 0 & 0\\
0 & 0 & 1
\end{pmatrix}\begin{pmatrix}
    p\left(1\middle|1\right) & p\left(1\middle|2\right) & p\left(1\middle|3\right)\\
    p\left(2\middle|1\right) & p\left(2\middle|2\right) & p\left(2\middle|3\right)\\
    p\left(3\middle|1\right) & p\left(3\middle|2\right) & p\left(3\middle|3\right)
    \end{pmatrix}.
\end{align}
This equality gives some indications on the relations between transition probabilities $p\left(i\middle|j\right)$. For transitions within deterministic states in the attractor of length 2, we have:
\begin{itemize}
\item $p\left(1\middle|1\right)=p\left(2\middle|2\right)$
\item $p\left(2\middle|1\right)=p\left(1\middle|2\right)$
\end{itemize}
There is no constraint on the transition probability $p\left(3\middle|3\right)$. Instead if we want to jump from the attractor of length 2 to the fixed point with a covariant influence, we can do so in such a way that $p\left(3\middle|1\right)=p\left(3\middle|2\right)$. Conversely, if we want to transition from the fixed point to the attractor of length 2 with a covariant influence, we can do so provided $p\left(1\middle|3\right)=p\left(2\middle|3\right)$.
Notice that a transition from a short attractor to a longer one is clearly was forbidden in the purely deterministic scenario described in Theorem~\ref{thm:necessary and sufficient}, so we can say that the presence of randomness activates some transitions and allows us to go against some of the rules prescribed for a dynamical system under a deterministic covariant influence.
\hfill
$\blacksquare$
\end{example}

Motivated by Example~\ref{ex:activated transitions}, we want to see how the presence of stochastic covariant influences affects the statement Theorem~\ref{thm:necessary and sufficient}. To this end, we focus on the columns of a stochastic matrix $F$ associated with a stochastic covariant influence $f$, where each column represents the evolution of a deterministic state (which can be transient or in an attractor) of the discrete dynamical system under $f$. 
This is stated in the following lemmas.
\begin{lemma}\label{lem:necessary d random}
Let $s$ and $s'$ be two deterministic states of a discrete dynamical
system $\left(S,\phi\right)$. If a transition occurs from $s$ to
$s'$ under a stochastic covariant influence, then $d'\leq d$.
\end{lemma}

\begin{proof}
We prove the statement by contradiction: we will assume $d'>d$, and
we show that a stochastic transition from $s$ to $s'$ is not possible,
i.e.\ $p\left(s'|s\right)=0$. Let $\boldsymbol{e}_i$ be the vector
of the canonical basis associated with the state~$s$. Then, if $F$
is the matrix of a covariant influence, the vector $\boldsymbol{p}:=F\boldsymbol{e}_i$
is the $i$th column of $F$, whose entries are the transition probabilities
$p\left(t|s\right)$, where $t$ is a deterministic state. Among them we find $p\left(s'|s\right)$.
Now,
\begin{equation}
\Phi^{\ell+d}\boldsymbol{e}_i
=\Phi^{d}\boldsymbol{e}_i
\end{equation},
where $\ell$ is the length of $s$, and~$\Phi$ is the matrix associated
with the generator of the dynamics $\phi$.

Similarly to Lemma~\ref{lem:transient necessary},
let us consider $\Phi^{\ell+d}F\boldsymbol{e}_i$:
\[
\Phi^{d}F\boldsymbol{e}_i=F\Phi^{d}\boldsymbol{e}_i=F\Phi^{\ell+d}\boldsymbol{e}_i=\Phi^{\ell+d}F\boldsymbol{e}_i,
\]
where we have exploited the covariance of $F$. Then we have 
\begin{equation}
\Phi^{\ell+d}\boldsymbol{p}=\Phi^{\ell}\Phi^{d}\boldsymbol{p}=\Phi^{d}\boldsymbol{p}.\label{eq:Phi^d p}
\end{equation}
It is useful to define $\boldsymbol{q}:=\Phi^{d}\boldsymbol{p}$. With
this in mind, Eq.~\eqref{eq:Phi^d p} reads $\Phi^{\ell}\boldsymbol{q}=\boldsymbol{q}$.
Now we use a similar technique to the one used in the proof of Proposition~\ref{prop:free random states}:
we partition the entries of $\boldsymbol{p}$ and $\boldsymbol{q}$
into blocks corresponding to the basins of attraction:
\begin{align*}
\boldsymbol{p} & =\boldsymbol{b}_1\oplus\ldots\oplus\boldsymbol{b}_k\\
\boldsymbol{q} & =\boldsymbol{c}_1\oplus\ldots\oplus\boldsymbol{c}_k.
\end{align*}
Let us focus on the basin of attraction of the state $s'$, which,
without loss of generality, we can assume to be $\boldsymbol{b}_1$
and $\boldsymbol{c}_1$, respectively.

As done in Proposition~\ref{prop:free random states},
we divide $\boldsymbol{b}_1$ and $\boldsymbol{c}_1$ into the
attractor part and the transient part. Recall that $\boldsymbol{q}$
is obtained by evolving $\boldsymbol{p}$ through $d$ time steps;
as a result, the entries of $\boldsymbol{c}_1$ can be obtained
by suitably moving the entries of $\boldsymbol{b}_1$ (including
$p\left(s'|s\right)$) by $d$ steps towards the attractor part. Note
that, since $d'>d$, the term $p\left(s'|s\right)$ will be in the
transient part of $\boldsymbol{c}_1$.
\[
\boldsymbol{c}_1=\begin{pmatrix}\textrm{attractor}\\
\hline \textrm{transient}
\end{pmatrix}=\begin{pmatrix}\textrm{attractor}\\
\hline \vdots\\
p\left(s'|s\right)+\cdots\\
\vdots
\end{pmatrix}
\]
Now, we know that $\Phi^{\ell}\boldsymbol{q}=\boldsymbol{q}$. This
implies, in particular, that $\Phi^{m\ell}\boldsymbol{q}=\boldsymbol{q}$,
for any $m\in\mathbb{N}$. By taking~$m$ large enough, we can make
the transient part of $\boldsymbol{c}_1$ vanish (cf.\ Proposition~\ref{prop:free random states}):
\[
\Phi^{m\ell}\boldsymbol{c}_1=\begin{pmatrix}\textrm{attractor}\\
\hline 0\\
\vdots
\end{pmatrix}=\begin{pmatrix}\textrm{attractor}\\
\hline \vdots\\
p\left(s'|s\right)+\cdots\\
\vdots
\end{pmatrix}=\boldsymbol{c}_1.
\]
As all the terms in $\boldsymbol{c}_1$ are non-negative, this equality
implies that $p\left(s'|s\right)=0$, which contradicts the hypothesis.
\end{proof}
\noindent
This lemma tells us that randomness in the covariant influence does not allow us to overcome the rule $d'\leq d$ we derived for the fully deterministic case.

\begin{lemma}
Let $s$ and $s'$ be two deterministic states of a discrete dynamical
system $\left(S,\phi\right)$. If a transition occurs from $s$ to
$s'$ under a stochastic covariant influence, then
\[
a\left(\phi^{n}\left(s'\right)\right)\geq a\left(\phi^{n}\left(s\right)\right),
\]
for every $n\in\mathbb{N}$.
\end{lemma}

\begin{proof}
Again, we prove the lemma by contradiction: we will assume there exists
an $n\in\mathbb{N}$ for which $a\left(\phi^{n}\left(s'\right)\right)<a\left(\phi^{n}\left(s\right)\right)$
and we show that a stochastic transition from $s$ to $s'$ is not
possible, i.e.\ $p\left(s'|s\right)=0$. Note that this inequality
implies that $\phi^{n}\left(s'\right)$ is a transient state, which
also means that $s'$ is transient. If $s$ is an attractor state,
then a transition from $s$ to $s'$ cannot occur due to Lemma~\ref{lem:necessary d random},
because the transient progeny of $s'$ is strictly greater than the
transient progeny of~$s$.

We are left with the case when $s$ and $s'$ are both transient.
Let $\boldsymbol{e}_i$ be the vector of the canonical basis associated
with the state~$s$. Then consider the vector $F\Phi^{n}\boldsymbol{e}_i$,
where $F$ is the matrix of a covariant influence. The entries of
$F\Phi^{n}\boldsymbol{e}_i$ are the transition probabilities from
$\phi^{n}\left(s\right)$, among which we find $p\left(\phi^{n}\left(s'\right)|\phi^{n}\left(s\right)\right)$.
As a first step, we will show that $p\left(\phi^{n}\left(s'\right)|\phi^{n}\left(s\right)\right)=0$.
Now, recall we can always find a deterministic state $s_{n}^{*}$
such that $\phi^{n}\left(s\right)=\phi^{a\left(\phi^{n}\left(s\right)\right)}\left(s_{n}^{*}\right)$.
If $\boldsymbol{e}_{i_{n}^{*}}$ is the vector of the canonical basis
associated with the state $s_{n}^{*}$, then we can write $\Phi^{n}\boldsymbol{e}_i=\Phi^{a\left(\phi^{n}\left(s\right)\right)}\boldsymbol{e}_{i_{n}^{*}}$.
Then,
\[
F\Phi^{n}\boldsymbol{e}_i=F\Phi^{a\left(\phi^{n}\left(s\right)\right)}\boldsymbol{e}_{i_{n}^{*}}=\Phi^{a\left(\phi^{n}\left(s\right)\right)}F\boldsymbol{e}_{i_{n}^{*}},
\]
where we have used the covariance of $F$. Now, as done in the proof
of Lemma~\ref{lem:necessary d random}, we partition $F\Phi^{n}\boldsymbol{e}_i$
and $F\boldsymbol{e}_{i_{n}^{*}}$ into basins of attraction.
\begin{align*}
F\Phi^{n}\boldsymbol{e}_i & =\boldsymbol{b}_1\oplus\ldots\oplus\boldsymbol{b}_k\\
F\boldsymbol{e}_{i_{n}^{*}} & =\boldsymbol{c}_1\oplus\ldots\oplus\boldsymbol{c}_k.
\end{align*}
Let us focus on the basin of attraction of the state $s'$, which,
without loss of generality, we can assume to be $\boldsymbol{b}_1$
and $\boldsymbol{c}_1$, respectively. Specifically, let $s_{n}'^{*}$
be the state such that $\phi^{n}\left(s'\right)=\phi^{a\left(\phi^{n}\left(s'\right)\right)}\left(s_{n}'^{*}\right)$,
which is in the same basin of attraction as $\phi^{n}\left(s'\right)$ and~$s'$.
Let us list the entries of $\boldsymbol{b}_1$ and $\boldsymbol{c}_1$
starting from $s_{n}'^{*}$, then moving to $\phi\left(s_{n}'^{*}\right)$,
and so on to the other successors of $s'^{*}$, until we come to $\phi^{n}\left(s'\right)$
after $a\left(\phi^{n}\left(s'\right)\right)$ steps. This means that
the entry relative to $\phi^{n}\left(s'\right)$ is located $a\left(\phi^{n}\left(s'\right)\right)$
below the first row:
\[
\boldsymbol{b}_1=\begin{pmatrix}p\left(s_{n}'^{*}|\phi^{n}\left(s\right)\right)\\
\vdots\\
p\left(\phi^{n}\left(s'\right)|\phi^{n}\left(s\right)\right)\\
\vdots
\end{pmatrix}\qquad\boldsymbol{c}_1=\begin{pmatrix}p\left(s_{n}'^{*}|s_{n}^{*}\right)\\
\vdots\\
p\left(\phi^{n}\left(s'\right)|s_{n}^{*}\right)\\
\vdots
\end{pmatrix}.
\]
Now, in the equality $\boldsymbol{b}_1=\Phi^{a\left(\phi^{n}\left(s\right)\right)}\boldsymbol{c}_1$,
the first $a\left(\phi^{n}\left(s'\right)\right)+1$ entries of the
vector at the right-hand side are zero:
\[
\boldsymbol{b}_1=\begin{pmatrix}p\left(s_{n}'^{*}|\phi^{n}\left(s\right)\right)\\
\vdots\\
p\left(\phi^{n}\left(s'\right)|\phi^{n}\left(s\right)\right)\\
\vdots
\end{pmatrix}=\begin{pmatrix}0\\
\vdots\\
0\\
\hline \textrm{non-zero}
\end{pmatrix}=\Phi^{a\left(\phi^{n}\left(s\right)\right)}\boldsymbol{c}_1.
\]
The reason for the presence of $a\left(\phi^{n}\left(s'\right)\right)+1$
zeros is that $s_{n}'^{*}$ is a farthest predecessor of $\phi^{n}\left(s'\right)$,
so at each step (i.e.\ at each power of~$\Phi$) we progressively
remove the transition probabilities associated with the successors
of $s_{n}'^{*}$. Therefore, $\Phi$ removes the probability associated
with $s_{n}'^{*}$, $\Phi^{2}$ removes the probability associated
with $\phi\left(s_{n}'^{*}\right)$, and so on, until we come to $\Phi^{a\left(\phi^{n}\left(s'\right)\right)+1}$,
which removes the probability associated with $\phi^{a\left(\phi^{n}\left(s'\right)\right)}\left(s_{n}'^{*}\right)=\phi^{n}\left(s'\right)$.
The fact that $a\left(\phi^{n}\left(s\right)\right)>a\left(\phi^{n}\left(s'\right)\right)$
guarantees that we are always able to remove the probability associated
with $\phi^{n}\left(s'\right)$. Therefore, we conclude that $p\left(\phi^{n}\left(s'\right)|\phi^{n}\left(s\right)\right)=0$.

Now we will show that $p\left(\phi^{n}\left(s'\right)|\phi^{n}\left(s\right)\right)=0$
implies $p\left(s'|s\right)=0$. Now, $p\left(\phi^{n}\left(s'\right)|\phi^{n}\left(s\right)\right)$
is an entry of the vector $F\Phi^{n}\boldsymbol{e}_i$. By covariance,
$F\Phi^{n}\boldsymbol{e}_i=\Phi^{n}F\boldsymbol{e}_i$. Let us
partition $F\boldsymbol{e}_i$ into basins of attraction:
\[
F\boldsymbol{e}_i=\boldsymbol{d}_1\oplus\ldots\oplus\boldsymbol{d}_k,
\]
where we assume that $\boldsymbol{d}_1$ is the basin of attraction
of~$s'$. This time, we list the entries of $\boldsymbol{b}_1$
and $\boldsymbol{d}_1$ starting from $s'^{*}$, where $s'^{*}$
is a deterministic state such that $s'=\phi^{a\left(s'\right)}\left(s'^{*}\right)$,
similarly to what we have done above. Then we have $\boldsymbol{b}_1=\Phi^{n}\boldsymbol{d}_1$.
In this way, due to the presence of $\Phi^n$, the term arising from $p\left(s'|s\right)$ will be in
the $n+a\left(s'\right)+1$ entry of $\Phi^{n}\boldsymbol{d}_1$.
On the other hand, the $n+a\left(s'\right)+1$ entry of $\boldsymbol{b}_1$
is
\begin{align*}
p\left(\phi^{n+a\left(s'\right)}\left(s'^{*}\right)|\phi^{n}\left(s\right)\right) & =p\left(\phi^{n}\circ\phi^{a\left(s'\right)}\left(s'^{*}\right)|\phi^{n}\left(s\right)\right)\\
 & =p\left(\phi^{n}\left(s'\right)|\phi^{n}\left(s\right)\right)\\
 & =0.
\end{align*}
Then the $n+a\left(s'\right)+1$ entry of $\boldsymbol{d}_1$ must
vanish; because it is a sum of non-negative terms (including $p\left(s'|s\right)$),
each term must vanish, from which we finally obtain that $p\left(s'|s\right)=0$.
\end{proof}

We summarize the previous lemmas with the following theorem.
\begin{theorem}
\label{thm:stochastic conversions} Let~$s$ and~$s'$ be two deterministic
states of a discrete dynamical system $\left(S,\phi\right)$. Then a transition from~$s$ to $s'$ is allowed under stochastic covariant influences only if
\begin{align}
\label{eq:condcovinflstoch}
&d'\leq d,\\
&a\left(\phi^n\left(s'\right)\right)\geq a\left(\phi^n\left(s\right)\right)\label{eq:ancestry condition stoch}
\end{align}
for $n\in \mathbb{N}$.
\end{theorem}

Notice that the divisibility constraint~\eqref{eq:condcovinfl} is no longer present in the statement of Theorem~\ref{thm:stochastic conversions}. Indeed, Example~\ref{ex:activated transitions} showed a concrete case where stochastic covariant influences overcome the divisibility constraint coming from deterministic covariant influences.

\section{Application to discrete logistic map}
\label{sec:app}
We now investigate how our results translate into a popular case of discrete dynamical system,
specifically, the discrete logistic map (DLM)~\cite{May1976}, even though it does not have a finite number of states, in contrast with the discrete dynamical systems we have studied so far.
The DLM
originated in the context of studying biological systems not reaching a steady state,
and the logistic map shows,
within a simple mathematical model,
not only is a steady state not achieved for some parameters but also that chaos emerges,
that a geometric convergence to chaos via period doubling occurs,
and period-three cycles~\cite{May1976,Logistic-bifurcation}.

\subsection{Mathematical model}
\label{subsec:mathmod}
We employ the standard mathematical description for the DLM~\cite{May1976}. We can think of the DLM as representing the evolution of a certain population. To achieve a compact treatment of the DLM, it is convenient to phrase it by introducing a variable~$x$ representing a fraction of the population with respect to the saturation level, i.e.\ the maximum population for the system. 
The DLM equation is
\begin{equation}
\label{eq:DLM}
x\gets rx(1-x),\;
r\in[0,4],\,
x\in[0,1]
\end{equation}
for~$\gets$ referring to the recursive replacement of~$x$
by its new value at each successive generation. Here~$r$ can be interpreted as a fertility rate.
For sufficiently low~$r$,
i.e., $r\leq 1$,
this recursive map yields the asymptotic value~$x_\infty=0$,
which denotes the steady-state solution and extinction.
However, larger~$r$
delivers more complicated solutions in the asymptotic long-time limit.

For $r>1$,
the DLM~(\ref{eq:DLM})
exhibits the phenomenon of period doubling until the onset of chaos at $r_\text{ch}\approx3.56995$.
Increasing~$r$
leads to bifurcating the steady state,
first into period~2,
then period~4 and so on
(powers of two).
The spacing between doubling,
with respect to~$r$,
asymptotically approaches a geometric progression with this progression quantified by Feigenbaum's number~$4.66920$.

The onset of chaos starts at~$r_\text{ch}$ and is characterized by exponentially sensitive dependence on initial conditions.
Although chaos dominates for $r\gtrapprox r_\text{ch}$,
these chaotic regions are punctuated by islands of stability including a period-three cycle,
whose existence actually guarantees the presence of chaos~\cite{LY75}.

We introduce an influence by modifying the recursion relation~\eqref{eq:DLM} to the form 
\begin{equation}
\label{eq:covinflrecursion}
x\gets rx(1-x)+f(x),
\end{equation}
where $f(x)$ is any real function of $x$. Introducing our notation, as $\phi$ evolves the discrete dynamical system for one time step, in the case of the DLM~\eqref{eq:DLM},
we have that $\phi\left(x\right)=rx\left(1-x\right)$.
Instead, the influence $f$ is given by the additive term added to the DLM. Imposing covariance amounts to imposing
$f\circ \phi = \phi \circ f$. We note that, if $f$ is a covariant influence, then $f^n$ is also a covariant influence, where by $f^n$ we mean $\underbrace{f\circ \dots \circ f}_{n \textrm{ times}}$. This can be easily proven by induction. To find examples of covariant influences from which to take powers and generate other covariant influences, we start by examining the case of a quadratic influence, i.e.\ of .e. $f(x)$ is of the form
\begin{equation}
    f\left(x\right)=ax^2 + bx +c,
\end{equation}
where $a,b,c\in\mathbb{R}$.

\subsection{Covariant influence for the DLM}
Now we explain how to impose covariance on the DLM, and we describe the situations that arise in this setting. Specifically, we will see that an influence of the form $f\left(x\right)=ax^2 + bx +c$ is covariant only for a very restricted set of values for its parameters. Moreover, we see that covariant polynomial influences of higher degrees can be obtained by taking iterates of the quadratic covariant term.

To impose covariance, we need to calculate $f\circ \phi$ and $\phi \circ f$ separately and then impose equality between these two expressions. Specifically,
\[
f\circ \phi = a\left(rx\left(1-x\right)\right)^2+brx\left(1-x\right)+c
\]
and
\[
\phi \circ f=r\left(ax^2+bx+c\right)\left(1-ax^2-bx-c\right).
\]
Imposing $f\circ \phi = \phi \circ f$ we get that $a,b,c$ can only take the following values.
We find that~$c$ must always be zero,
meaning that the covariant influence is forbidden from applying a constant shift.
Given that~$c\equiv0$,
two solutions arise:
either $a=0$ and $b=1$,
in which case the influence is strictly linear in~$x$,
or $a=-r$ and $b=r$,
in which case the influence is strictly quadratic.

Given that the only covariant influence of polynomial degree one is $x\gets x$,
which is an identity map,
and the case of polynomial degree two yields a rescaled logistic map,
we investigate a higher polynomial degree. In this case, a natural choice of covariant influence is obtained by considering the $n$th iterate of the logistic map, i.e.\ $f=\phi^n$.

\subsection{Discussion of `influenced' logistic map}
\label{subsec:discinfllogmap}

The linear case is not a DLM, as the recursive relation is
\begin{equation}
\label{eq:linearcase}
x \gets rx\left(1-x\right)+x,
\end{equation}
which is not of the form in Eq.~\eqref{eq:DLM}. Imposing the constraint that $x\in\left[0,1\right]$ at all generations, as it should be, given that $x$ represents a population fraction, we find that the new recursive relation is well-posed if and only if $r\in\left[0,1\right]$. In this regime, the uninfluenced DLM yields population extinction for $r\in\left[0,1\right]$. The linear term can be interpreted as proportional immigration term that is proportional to the population fraction. An easy analysis of the convergence of the new recursive relation yields that the presence of the linear migration terms always leads to population saturation (unless, of course, if the initial population is zero). With this interpretation, we can understand why the DLM with the covariant linear term is well posed only for $r\in\left[0,1\right]$: the population decay, as prescribed by the uninfluenced DLM, balances the proportional immigration kick, making $x$ not exceed the maximum population.

After straightforward calculations, in the case of a quadratic covariant influence, the conditions
\begin{equation}
a=-r,\,b=r,\,c=0,
\end{equation}
make Eq.~\eqref{eq:covinflrecursion} into the form \begin{equation}
x\gets 2rx(1-x),
\end{equation}
which just yields a rescaling of the original DLM by a factor 2. This leads to a redefinition of $r\in\left[0,2\right]$ in order to keep $x\in\left[0,1\right]$.
In this case, the uninfluenced part $\phi$ and the influence $f$ coincide, so $\phi$ and $f$ actually describe the same evolution. We can interpret this situation as a splitting of the original population obeying the DLM into two halves. As~$x$ is still constrained to be in $\left[0,1\right]$ and~$r$ can be interpreted as the fertility rate of the population~\cite{May1976}, splitting the population into two means that both~$x$ and~$r$ must be renormalized against the reduced size of the population. In particular, after the splitting into two halves, $r$ becomes $r':=\nicefrac{r}{2}$. Therefore, when considering the evolution of the whole population, it is natural that the fertility rate will be twice the renormalized fertility rate of each half; in symbols $r=2r'$.

In the case of higher-degree polynomial influences the recursive relation is of the form
\begin{equation}
\label{eq:higherdegree}
    x \gets rx\left(1-x\right)+\bigcirc^nrx\left(1-x\right)
\end{equation}
with~$\bigcirc^n$ the $n$th-order composition of Eq.~\eqref{eq:DLM} with itself~$n$ times.
We can interpret this expression as the presence of immigration that evolves exactly as the original system but $n$ times as fast. Notice that this general case subsumes the two previous cases, for $n=0$ and $n=1$, respectively. In this approach, we need to take care that $x\in\left[0,1\right]$, being $x$ a population fraction. If $r$ is small, the uninfluenced evolution will lead to a decay in fertility rate, and the immigration term, which evolves $n$ times as fast, will produce an even smaller contribution, thus ensuring that $x$ is always in the interval $\left[0,1\right]$.  

Suppose an initial population of $x=\nicefrac12$,
which means half the saturation value.
If $r=\nicefrac34$,
then the sequence of population fractions is $\nicefrac12\to\nicefrac3{16}$.
With proportional immigration~\eqref{eq:linearcase},
we instead obtain $\nicefrac12\to\nicefrac{11}{16}$, which highlights the role of immigration in reaching a higher population fraction. In this case, the population fraction will reach a steady state of 1 due to the contribution from immigration.
In the quadratic case, assuming again $r=\nicefrac34$, the sequence of population fractions is $\nicefrac12\to\nicefrac3{8}$, which can be immediately seen as twice as fast the uninfluenced logistic evolution, approaching the 0 steady state.
As an instance of higher-degree covariant influence, we consider the case $n=2$. In this case, the immigration term produces a contribution that decays faster than the decay induced by the uninfluenced logistic map. Specifically, starting from $\nicefrac12$ we reach $\nicefrac{309}{1024}$ in the next iteration. If we continue the iterations, we see that the population fraction approaches the steady state of 0.1975,
which corresponds to the intersection between $x\gets x$ and the function on the right-hand side of~$\gets$ in Eq.~\eqref{eq:higherdegree} for $n=2$.
Thus, our numerical illustration shows that the evolution leads to a steady-state population for $r=\nicefrac34$ in each of the $n\in\{0,1,2\}$ instances,
and this steady-state population agrees with the theory expounded above.

\section{Discussion}
\label{sec:discussion}
We study discrete dynamical systems over finite sets, which arise in several situations in different disciplines. In discrete dynamical systems states evolve in discrete steps according to a deterministic rule. Since we assume that the set of states is finite, after a finite time, the states must repeat themselves, ending up in a cycle. Discrete dynamical systems come in two flavors. The first is purely deterministic, where we know the initial state perfectly, and we are able to predict the evolution of each state with certainty. The second incorporates some randomness, as we do not know the initial state precisely, but only a probability distribution over possible initial states. Despite the randomness in the states, the evolution remains deterministic: if we know the initial state perfectly, we will still be able to predict its evolution with certainty.

Randomness in the initial state arises to model our ignorance about the exact initial state. This happens because the precision with which we probe a system is normally limited, and an exact description of the initial state is thus not achievable or highly impractical. In such situations, we can still describe the system by shifting from the evolution of states to the evolution of probability distributions over states.

As no systems are truly isolated, we incorporate the environment in the description of a discrete dynamical system, where the environment is intended as everything that is outside the dynamical system. In particular, the environment disrupts the normal evolution of the discrete dynamical system, and it can do so in two ways: either deterministically, or randomly. In some cases, the external influence occurs over a longer time scale than the time scale of evolution of a dynamical system. In such situations, the dynamical system has time to adapt to the external influence, so that the influence itself does not become too disruptive, but it interplays nicely with the dynamic of the system. Such an influence is called covariant, and it is the subject of our analysis.

We examine the effect of a covariant influence on a discrete dynamical system with multiple attractors and an environment enacting a deterministic covariant influence. Without the covariant influence, any state will eventually reach an attractor. Turning on the covariant influence allows hopping between different attractors, where we show that hopping from one attractor to another one is possible if and only if the length of the final attractor divides the length of the initial one. Thus, covariant influences shorten the length of attractors, so that the system spontaneously reaches a state of ``minimal evolution''. From the perspective of resource theories, a resource is something that goes against the spontaneous evolution induced by covariant influences. For this reason, the resource into play in this setting is called ``evolutionarity''.
From a biological point of view, this means that a system could be too rigid to adapt to the changes coming from the environment.

On the other hand, if the environment acts on a dynamical system in a complex way, its covariant influence must be described by a stochastic function. The effect of such an influence is always to allow hopping between different attractors, but unlike its deterministic version, now every deterministic state of the system can jump probabilistically to another deterministic state, where each transition is weighted by a certain probability, in contrast to what happened for deterministic influences. As usual, in the long run, the effect of the stochastic covariant influence is to send a deterministic state to a state in an attractor, where all states in the attractor have the same probability to be reached. Here, again, a resource is something that goes against the spontaneous evolution induced by covariant influences. Hence, ``non-attractorness'' is the resource that steers the probability distribution away from either being only supported on attractors or not being uniform on the states of an attractor. Such a resource makes deterministic states abandon their attractor and jump either temporarily to a transient state, or to a state of another attractor; viz.\ they behave against what covariant influences enforce.  This time, there are no constraints on the length of the target attractor, which means that covariant influences can also increase the length of attractor.  From a biological point of view, this means that a system now can be adaptable to diverse changes induced by the environment.

The results of this article are derived using the general and powerful framework of resource theories, which has been used extensively in quantum information and foundations. For the first time, we exported resource theories outside the domain of physics. Besides the conceptual clarity resource theories brought to the study of discrete dynamical systems and their influences, resource theories also inspired our questions and informed our approach to finding an answer to those questions. For example, the question of how covariant influences affect the evolution of a discrete dynamical systems can be phrased as the conversion problem in resource theories, which in turn can be solved using resource monotones. Such monotones give rise precisely to the constraints we found for the behavior of covariant influences discussed above.

\section{Conclusions}
\label{sec:conclusions}
We introduced the notion of covariant influence, which is the backbone of our analysis. The notion of covariant influence is expressed in terms of the commutation of the influence $f$ with the influenced dynamics $\phi$.
The uninfluenced dynamics is in general not reversible.

We developed the theory of discrete dynamical systems under covariant influences, giving a  mathematical characterization of which transitions can take place in a discrete dynamical system due to the presence of a covariant influence. Specifically, we analyzed two situations: one where there is no randomness, and the other where the covariant influence is stochastic. These two situations are qualitatively and quantitatively different. For the case of the lack of randomness, we derived necessary and sufficient conditions for the evolution of states under covariant influences (Theorem~\ref{thm:necessary and sufficient}). In this case, the covariant influence allows jumping between attractors whose final length divides the initial length. On the other hand, we showed that randomness permits transitions that violate this divisibility rule, and we were able to give some constraints on the entries of a generic stochastic matrix representing a covariant influence. However, understanding if a stochastic state can be converted into another with a covariant influence is in general a rather complicated problem, which can be formulated as a linear decision problem.

The underpinning of our results is the framework of resource theories, which is a powerful tool developed in quantum information. In this work, for the first time, we exported it to another area of science, and we showed how it can be applied to the study of dynamical systems, which are widespread in science. The application of resource theories is the mathematical foundation that allowed us to split the evolution of 
a discrete dynamical system into an influence and an uninfluenced part, in the same spirit one does perturbation theory, but without the restriction that the influence be a small contribution.

Resource theories guarantee the solidity of our results, as they are expressed with a logical and rigorous mathematical framework that can be formalized using category theory. This ensures the validity of our findings for all discrete dynamical systems, as long as the covariance condition holds. The level of generality we obtained is particularly valuable, in that it allows one to understand the behavior of any discrete dynamical system under a covariant influence, without needing to know the exact details of the evolution of that dynamical system. Our results are in the form of laws that can be tested in concrete dynamical systems, paving the way for further analysis, even from an experimental perspective.

There are open questions left for future work. One direction involves formulating a generalized version of the covariant condition, perhaps involving powers of $\phi$, and analyzing its consequences. Another direction further explores the action of stochastic covariant influences. For example, knowing the dynamical graph of a discrete dynamical system, we would like to be able to write the form of a stochastic covariant matrix immediately, just by looking at the dynamical graph. In particular, it is interesting to characterise the transition probabilities between states in attractors of different lengths. Indeed, we saw in the deterministic case that transitions under a covariant influence are permitted only if the length of the final attractor divides the length of the initial one. In the stochastic case, on the other hand, such transitions are allowed, but we may still observe some special behaviour in the transition probabilities when the final state is in an attractor whose length divides the length of attractor in which the initial state is. A third direction involves studying applications of our framework based on covariant influences to concrete settings, such as genetic regulatory networks.

\acknowledgments
We acknowledge the traditional owners of the land on which this work was undertaken at the University of Calgary: the Treaty 7 First Nations.
CMS acknowledges support from the Pacific Institute for the Mathematical
Sciences (PIMS) and from a Faculty of Science Grand Challenge award
at the University of Calgary. CMS thanks Stefano Gogioso for insightful
conversations,
and we appreciate Eduardo Miguel Mart\'{i}nez Garc\'{i}a for computing iterations of modified logistic map.
This work has been supported by NSERC.
\appendix
\section{Conversion problem: evolutionarity between different dynamical systems}\label{app:different dynamical}
In~\S\ref{subsec:conversion} we presented necessary and sufficient conditions for the conversion of a state $s$ into of a state $s'$ of the same dynamical system $\left(S,\phi\right)$. Now we generalize this situation by considering the conversion of a state $s$ of $\left(S,\phi\right)$ into a state $s'$ of $\left(S',\phi'\right)$ under any covariant map $f:S\to S'$.

As to \emph{necessary} conditions, notice that the proofs of Lemmas~\ref{lem:transient necessary}  and~\ref{lem:ancestry necessary} can be easily extended to the case when $s\in S$ and $s'\in S'$: it is enough to insert a prime where needed; therefore Lemmas~\ref{lem:transient necessary}  and~\ref{lem:ancestry necessary} still hold when there are two dynamical systems involved. In fact, in some cases, we can use Lemma~\ref{lem:transient necessary} to show that \emph{no} free operations exist
between  $\left(S,\phi\right)$ and $\left(S',\phi'\right)$, as shown in the following
example.
\begin{example}
\label{exa:no free}Suppose $\left(S,\phi\right)$ is a discrete dynamical system
with only one attractor of length 3, and $\left(S',\phi'\right)$
is a discrete dynamical system with only one attractor of length 2. If there
were free operations $f:S\rightarrow S'$,
we would map states in the only attractor of $\left(S,\phi\right)$
to states in the only attractor of $\left(S',\phi'\right)$, but the length
of the attractor of $\left(S',\phi'\right)$ does \emph{not}
divide the length of the attractor of $\left(S,\phi\right)$.
Therefore, no covariant operations from~$S$ to
$S'$ exist, in agreement with Lemma~\ref{lem:transient necessary}. For the same reason, neither
do free operations from~$S'$ to $S$ exist.

Note that the possible presence of transient states associated with
the attractors does not affect this result at all. This shows how
strong the condition on divisibility of Lemma~\ref{lem:transient necessary} is.
\hfill
$\blacksquare$
\end{example}

 Now, let us turn to \emph{sufficient} conditions when $\left(S,\phi\right)$ is not necessarily equal to $\left(S,\phi'\right)$.
In this case, given two states~$s\in S$  
and~$s'\in S'$, the two conditions $d'\leq d$
and $\ell'\mid\ell$ in Theorem~\ref{thm:necessary and sufficient} no longer suffice as shown in the following
example.
\begin{example}
\label{ex:differentnetworks}
 Let $\left(S,\phi\right)$ be a discrete dynamical system with one attractor
of length 3 and one of length 4. Let $\left(S',\phi'\right)$
be a discrete dynamical system with only one attractor of length 2. Let~$s$ be a state in the attractor of
length 4 of $\left(S,\phi\right)$, and let~$s'$ be a state
in the attractor of length 2 of $\left(S',\phi'\right)$. $s$
and~$s'$ respect the conditions $d'\leq d$ and $\ell'\mid\ell$,
but a transition is not possible between them because covariant maps do not
exist between $\left(S,\phi\right)$ and $\left(S',\phi'\right)$ because of incommensurate attractor lengths.
Indeed, if they did exist, we must map the attractor of length 3 in $\left(S,\phi\right)$
to an attractor with a length that divides 3, but such an attractor does not exist in $\left(S',\phi'\right)$.
\end{example}

 Example~\ref{ex:differentnetworks} shows that,
when $\left(S,\phi\right)\neq\left(S',\phi'\right)$,
the conversion problem between~$s$ and~$s'$ does not depend
only on~$s$ and $s'$, but involves the other deterministic states. The reason is the lack of covariant maps from
one dynamical system to the other. Here we provide a necessary and sufficient
condition for the existence of covariant maps between two different
discrete dynamical systems.
\begin{proposition}
\label{prop:existence free operations} Given two discrete dynamical systems,
$\left(S,\phi\right)$ and $\left(S',\phi'\right)$,
there exist covariant maps $f:S\rightarrow S'$
iff,
for every attractor of $\left(S,\phi\right)$
of length $\ell$,
there exists an attractor of $\left(S',\phi'\right)$
of length $\ell'\mid\ell$.
\end{proposition}
\begin{proof}
 Necessity follows from Lemma~\ref{lem:transient necessary}.
To prove sufficiency, we need to show that, under the hypotheses of
the theorem, it is possible to construct a covariant $f:S\rightarrow S'$.
To this end, for every attractor of length $\ell$ in $\left(S,\phi\right)$,
let us pick an attractor state $s$, and a state~$s'$ in an attractor
of $\left(S',\phi'\right)$ whose length $\ell'$ divides
$\ell$. Then map all the states in the basin of attraction of
$s$ to states in the basin of attraction of $s'$, following the
procedure outlined in the proof of Lemma~\ref{lem:sufficient1} (in this case $d=0$). This yields a covariant map on all the
states in the basin of attraction of~$s$. Repeating this procedure
for all the attractors of $\left(S,\phi\right)$, we can construct
a covariant map~$f$ on all states of $\left(S,\phi\right)$.
\end{proof}
Now we can provide necessary and sufficient conditions for the conversion of a deterministic state in $\left(S,\phi\right)$ into a state of $\left(S',\phi'\right)$ with covariant maps.
\begin{theorem}
\label{thm:necessary and sufficient general} Let~$s$ be a deterministic
state of a discrete dynamical system $\left(S,\phi\right)$ and $s'$
be a deterministic state of a discrete dynamical system $\left(S',\phi'\right)$.
Then there exists a covariant map $f:S \rightarrow S'$
converting~$s$ into~$s'$ iff
\begin{enumerate}
\item there exist covariant  maps from $\left(S,\phi\right)$ to
$\left(S',\phi'\right)$;
\item $\ell'\mid\ell$;
\item $d'\leq d$;
\item $a\left(\phi'^n\left(s'\right)\right)\geq a\left(\phi^n\left(s\right)\right)$
for $n=0,\ldots,d'-1$.
\end{enumerate}
\end{theorem}
\begin{proof}
 Again, necessity was proved in Lemmas~\ref{lem:transient necessary}
and~\ref{lem:ancestry necessary}. Sufficiency follows from combining the proofs of Theorem~\ref{thm:necessary and sufficient}
and Prop.~\ref{prop:existence free operations}. Indeed,
under the hypotheses of the theorem, following the proof of Theorem~\ref{thm:necessary and sufficient}
we can map the basin of attraction of~$s$ to the basin of attraction
of~$s'$ in a covariant way. Regarding the other basins of attraction
in $\left(S,\phi\right)$, following the proof of Prop.~\ref{prop:existence free operations},
we can covariantly map every basin of attraction of $\left(S,\phi\right)$
to a basin of attraction of $\left(S',\phi'\right)$.
\end{proof}

 Notice that in Theorem~\ref{thm:necessary and sufficient general} we have one additional condition than in Theorem~\ref{thm:necessary and sufficient}: the existence of covariant maps from  from $\left(S,\phi\right)$ to $\left(S',\phi'\right)$. This condition is not present in Theorem~\ref{thm:necessary and sufficient}. Indeed, if $\left(S,\phi\right)=\left(S',\phi'\right)$, the existence of covariant maps is always guaranteed. To understand why, note that,
if $\left(S,\phi\right)=\left(S',\phi'\right)$, then,
for every attractor of length $\ell$,
we can always find an attractor
of length $\ell'\mid\ell$: it is enough to take the same attractor.
By Prop.~\ref{prop:existence free operations}, this is enough
to guarantee the existence of covariant maps. This is why Theorem~\ref{thm:necessary and sufficient}
does not mention this check at all.

\section{Linear programming}
\label{app:linear programming}

In this appendix, we show that, in the presence of randomness and stochastic influences, we can recast the conversion problem defined in~\S\ref{subsec:nonattractorness}. Here, a square matrix~$\Phi$, representing the dynamics, is given by the specification of the dynamical system, and we want to know if, for a pair of probability vectors $\bm{p}$ and $\bm{q}$, there exists a stochastic matrix $F$ such that $F\bm{p}=\bm{q}$ and $\left[F,\Phi\right]=0$.

Let~$\Phi$ be the dynamical matrix of a dynamical system $S$ with~$M$ states (see~\S\ref{subsec:DDS}).
As such, $\Phi$ can be expressed in terms of the generator of the dynamics $\phi:S\to S$ as
\begin{equation}
\Phi=\left(\bm{e}_{\phi\left(1\right)}\;\;\bm{e}_{\phi\left(2\right)}\;\;\cdots\;\;\bm{e}_{\phi\left(M\right)}\right).
\end{equation}
Let 
\begin{equation}
F=\sum_{i=1}^{M}\sum_{j=1}^{M}f_{ij}\bm{e}_i\bm{e}_j^T
\end{equation}
be an $M\times M$ column-stochastic matrix, and $T$ denote the transpose. Observe that the $j$th column of $F\Phi$ is given by $F\bm{e}_{\phi(j)}
$. On the other hand, the $j$th column of $\Phi F$ is given by $\sum_{k=1}^Mf_{kj}\bm{e}_{\phi\left(k\right)}$. Note that the $i$th component of $\sum_{k=1}^nf_{kj}\bm{e}_{\phi\left(k\right)}$ is given by 
$
\sum_{k\in\phi^{-1}(i)}f_{kj}
$, where $\phi^{-1}\left(i\right)$ denotes the pre-image of $s_i$. 

Hence, the condition $F\Phi=\Phi F$ is equivalent to
\begin{equation}
f_{i\phi(j)}=\sum_{k\in\phi^{-1}\left(i\right)}f_{kj},
\end{equation}
for every $j,k\in\left\{1,\dots,M\right\}$. This condition in turn can be expressed in terms of Hilbert-Schmidt inner products 
\begin{equation}
\mathrm{tr}\left[F G_{ij}^T\right]=0
\end{equation}
for every $i,j\in\left\{1,\dots,M\right\}$, where
\begin{equation}
G_{ij}=\bm{e}_i\bm{e}_{\phi\left(j\right)}^T-\sum_{k\in\phi^{-1}\left(i\right)}\bm{e}_k\bm{e}_j^T.
\end{equation}

The condition that the columns of $F$ sums to one can be expressed as 
\begin{equation}\label{eq:Fu=u}
F\bm{u}=\bm{u}, \qquad \bm{u}:=\begin{pmatrix}
1 \\ \vdots \\ 1 \end{pmatrix}.
\end{equation}
Define
\begin{equation}
    U_1:= \begin{pmatrix}
    1 & \cdots & 1 \\
    0 & \cdots & 0 \\
    \vdots &\ddots &\vdots 
    \end{pmatrix};
\end{equation}
the idea is to reshape the vector $\bm{u}$ into a matrix. For each $i>1$, define the matrix 
\begin{equation}
    U_i:=\begin{pmatrix}
    1 & \cdots & 1 \\
    0 & \cdots & 0 \\
    \vdots &\ddots &\vdots \\
    -1 & \cdots & -1 \\
    0 & \cdots & 0 \\
    \vdots &\ddots &\vdots \\
    \end{pmatrix},
\end{equation}
where the first row is $\bm{u}^T$, the $i$th row is $-\bm{u}^T$, and the rest of the rows are zeros. Then, we can express the condition~\eqref{eq:Fu=u} as 
\begin{equation}
\mathrm{tr}\left[FU_i^T\right]=\updelta_{i1}
\end{equation}
for all $i\in\left\{1,\dots,M\right\}$.

Let $\bm{p},\bm{q}\in\mathbf{P}^M\left(\mathbb{R}\right)$ be $M$-dimensional probability vectors of which we want to study the conversion under covariant influences. For each $i\in\left\{1,\dots,M\right\}$, define the matrix \begin{equation}
H_i:=\begin{pmatrix}
0 & \cdots &0 \\
\vdots & \ddots & \vdots \\
p_1-q_i & \cdots & p_M-q_i \\
\vdots & \ddots & \vdots \\
0 & \cdots &0 
\end{pmatrix},
\end{equation}
whose $i$th row is $\bm{p}^T-q_i\bm{u}^T$ and the rest of the rows are zeros. Then, we can express $F\bm{p}=\bm{q}$ as 
\begin{equation}
\mathrm{tr}\left[FH_i^T\right]=0
\end{equation}
for all $i\in\left\{1,\dots,M\right\}$.

We therefore arrive at the following SDP decision problem. Determine if there exists $F\geq 0$ (entry-wise) such that for all $i,j\in\left\{1,\dots,M\right\}$
\begin{enumerate}
\item $\mathrm{tr}\left[FU_i^T\right]=\updelta_{i1}$;
\item $\mathrm{tr}\left[F G_{ij}^T\right]=0$;
\item $\mathrm{tr}\left[FH_i^T\right]=0$.
\end{enumerate}

Note that,
if we remove the first condition $\mathrm{tr}\left[FU_i^T\right]=1$, we just need to assume in addition that $F$ is a non-zero matrix whose entries are non-negative. Therefore, there exists such an $F$ in the orthogonal complement $\mathcal{W}^\perp$ to the subspace
\begin{equation}
\mathcal{W}=\mathrm{span}_{\mathbb{R}}\left\{U_2,...,U_M,G_{11},...,G_{MM},H_1,....,H_M\right\},
\end{equation}
with respect to the Hilbert-Schmidt inner product.

Alternatively, let $\bm{f}\in\mathbb{R}^{M^2}$ be the vector obtained by reshaping  $F$. Similarly, let $\bm{g}_{ij}$, $\bm{u}_i$, and  $\bm{h}_j$ be the vectors obtained by reshaping $G_{ij}$, $U_i$, and $H_i$, respectively. With these notations,
the inner products above take the form of the dot products
\begin{enumerate}
\item $\bm{f}\cdot\bm{u}_i=\updelta_{i1}$;
\item $\bm{f}\cdot\bm{g}_{ij}=0$;
\item $\bm{f}\cdot\bm{h}_i=0$.
\end{enumerate}
Hence, the above conditions are equivalent to the existence of a non-trivial and non-negative solution to the homogeneous linear system of equations
\begin{equation}
A\bm{f}=0
\end{equation}
where 
\begin{equation}
A=\left(\bm{u}_2,...,\bm{u}_M,\bm{g}_{11},...,\bm{g}_{MM},\bm{h}_1,...,\bm{h}_M\right)^T
\end{equation}
is an $\left(M^2+2M-1\right)\times M^2$ matrix.

\bibliography{bibliography}

\end{document}